\newcommand*{\ARXIV}{}

\documentclass[journal,12pt,onecolumn,draftclsnofoot]{IEEEtran}
\usepackage[utf8]{inputenc} \usepackage[T1]{fontenc}    \usepackage{hyperref}       \usepackage{url}            \usepackage{booktabs}       \usepackage{amsfonts}       \usepackage{nicefrac}       \usepackage{microtype}      \usepackage{xcolor}         \usepackage{amssymb,amsfonts,amstext,mathrsfs}
\usepackage{graphics,latexsym,epsfig,psfrag,wrapfig,comment,paralist}
\usepackage{xspace}
\usepackage{bm,multirow}
\usepackage{booktabs}
\usepackage{mathdots}
\usepackage{bbold}
\usepackage{float}
\usepackage{centernot}
\usepackage{prettyref}
\usepackage{listings}
\usepackage{tikz}
\usepackage{pgfplots}
\usetikzlibrary{shapes}
\usetikzlibrary{positioning, arrows, matrix, calc, patterns, fit}
\usepackage{paralist}
\usepackage{array}
\usepackage{mdwmath}
\usepackage{multirow}
\usepackage{mdwtab}
\usepackage{eqparbox}
\usepackage{multicol}
\usepackage{amsfonts}
\usepackage{tikz}
\usepackage{multirow,bigstrut,threeparttable}
\usepackage{amsthm}
\usepackage{array}
\usepackage{bbm}
\usepackage{epstopdf}
\usepackage{mdwmath}
\usepackage{mdwtab}
\usepackage{dsfont}
\usepackage{eqparbox}
\usepackage{tikz}
\usepackage{latexsym}
\usepackage{amssymb}
\usepackage{enumitem}
\usepackage{bm}
\usepackage{graphicx}
\usepackage{mathrsfs}
\usepackage{epsfig}
\usepackage{psfrag}
\usepackage{setspace}
\usepackage{algorithm}
\usepackage{algpseudocode}
\usepackage{subcaption}
\usepackage{xstring}
\usepackage{stackengine}
\newcommand{\VEC}[1]{\boldsymbol{\mathbf{#1}}}
\pgfplotsset{compat=1.17}

\newtheorem{theorem}[]{Theorem}

\newtheorem{assumption}[]{Assumption}
\newtheorem{lemma}[]{Lemma}
\newtheorem{proposition}[]{Proposition}

\newcommand{\abs}[1]{\left \lvert {#1} \right \rvert}
\newcommand{\type}[1]{\mathrm{Type}\left({#1}\right)}

\newcommand{\norm}[1]{\left \lVert {#1} \right \rVert}

\newcommand{\eps}{\epsilon}
\newcommand{\Z}{\mathbb{Z}}
\newcommand{\R}{\mathbb{R}}
\newcommand{\E}{\mathbb{E}}
\newcommand{\inner}[2]{\left \langle #1, #2 \right \rangle}

\usepackage{amsbsy}
\usepackage[cmex10]{amsmath}

\DeclareMathOperator*{\argmin}{arg\,min} 
\usepackage[utf8]{inputenc} 
\usepackage[T1]{fontenc}
\usepackage{url}              \usepackage{cite}             

\usepackage[cmex10]{amsmath}  \interdisplaylinepenalty=10\usepackage{mleftright}       \mleftright                   

\usepackage{graphicx}         \usepackage{booktabs}

\begin{document}

\title{Efficiently Computing Sparse Fourier Transforms of $q$-ary Functions}
\ifdefined\ARXIV
\author{\IEEEauthorblockN{Yigit Efe Erginbas$^*$, Justin Singh Kang$^*$, Amirali Aghazadeh, \\ Kannan Ramchandran}
}
\else
\author{\IEEEauthorblockN{Anonymous Authors}
  \IEEEauthorblockA{Please do NOT provide authors' names and affiliations\\
    in the paper submitted for review, but keep this placeholder.\\
    ISIT23 follows a \textbf{double-blind reviewing policy}.}
}
\fi

\def\thefootnote{*}\footnotetext{These authors contributed equally to this work.}
\maketitle

\begin{abstract}
\ifdefined\ARXIV
\else
THIS PAPER IS ELIGIBLE FOR THE STUDENT PAPER AWARD.
\fi Fourier transformations of pseudo-Boolean functions are popular tools for analyzing functions of binary sequences. Real-world functions often have structures that manifest in a sparse Fourier transform, and previous works have shown that under the assumption of sparsity the transform can be computed efficiently. But what if we want to compute the Fourier transform of functions defined over a \(q\)-ary alphabet? These types of functions arise naturally in many areas including biology. A typical workaround is to encode the \(q\)-ary sequence in binary, however, this approach is computationally inefficient and fundamentally incompatible with the existing sparse Fourier transform techniques. Herein, we develop a sparse Fourier transform algorithm specifically for \(q\)-ary functions of length \(n\) sequences, dubbed \(q\)-SFT, which provably computes an \(S\)-sparse transform with vanishing error as \(q^n \rightarrow \infty\) in \(O(Sn)\) function evaluations and \(O(S n^2 \log q)\) computations, where \(S = q^{n\delta}\) for some \(\delta < 1\). Under certain assumptions, we show that for fixed \(q\), a robust version of \(q\)-SFT has a sample complexity of \(O(Sn^2)\) and a computational complexity of \(O(Sn^3)\) with the same asymptotic guarantees. We present numerical simulations on synthetic and real-world RNA data, demonstrating the scalability of \(q\)-SFT to massively high dimensional \(q\)-ary functions.

\end{abstract}

\section{Introduction}
Pseudo-Boolean functions \cite{Odonnell2021} are powerful tools for modeling systems across many fields including computer science \cite{bockmayr1991logic}, biology \cite{Kargupta2001} and game theory \cite{ Marichal2000}. \
In many of these applications, the Fourier transform, also known as the Walsh-Hadamard transform, plays an important role. In particular, a Pseudo-Boolean function $f : \Z_2^n \rightarrow \R$ can be represented as a real, multilinear polynomial in terms of its Fourier transform $F$, i.e.,
\begin{equation}
    f[\mathbf{m}] = \sum_{\mathbf{k} \in \Z_2^n} F[\mathbf{k}](-1)^{\inner{\mathbf{m}}{\mathbf{k}}}, \;\;\mathbf{m} \in \Z_2^n, 
\end{equation}
where $N = 2^n$ and $\Z_2=\{0,1\}$. The transform $F$ measures interactions between subsets of coordinates in ${\bf k}$ and is often more interpretable than $f$. In fact, Fourier techniques are now being used for explaining deep neural networks in real-world problems \cite{aghazadeh2021epistatic,ha2021adaptive}. It is known that $F$ can be computed in $O(N\log{N})$ computations using the seminal fast-Fourier transform (FFT) algorithm. 
This exponential dependence in $n$ makes computing $F$ impractical for even moderate values of $n$. In most applications, however, $F$ has some exploitable low-dimensional structure. By far the most commonly considered structure is sparsity of the Fourier transform~\cite{Mansour1994}. Compressed sensing theory shows that algorithms like LASSO~\cite{tibshirani1996regression} achieve the optimal sample complexity of $O(S\log(N/S)) $ in recovering an $S$-sparse Fourier transform $F$ (i.e., a function with only $S$ non-zero values). 
More recent \emph{sparse Fourier transform} algorithms~\cite{stobbe2012, Scheibler2013, Li2015, Amrollahi2019} demonstrate that the sparsity in $F$ can be exploited even further to significantly reduce the computational complexity of Fourier transforms. 
In particular, it is shown that a Boolean $S$-sparse Fourier transform can be computed in only $O(S\log^2N) = O(Sn^2)$ time complexity~\cite{Li2015}.

Though Pseudo-Boolean functions are a complete model for functions of binary sequences, in practice real data may involve $q$-ary sequences like DNA or proteins, which are defined over a grammar in a $q$-ary 
alphabet with $q>2$: $q=4$ for DNA, $q=20$ for proteins, and $q=4^\ell$ for the $\ell$-mer representation of DNA. To address the large and important class of problems with $q >2$, here we introduce \hyperref[alg]{$q$-SFT} for computing an $S$-sparse $q$-ary Fourier transform:
\begin{equation}\label{eq:q_fourier_transform}
    F[\mathbf{k}] = \frac{1}{N}\sum_{\mathbf{m} \in \Z_q^n} f[\mathbf{m}]\omega^{-\inner{\mathbf{m}}{\mathbf{k}}}, \;\;\mathbf{k} \in \Z_q^n,
    \vspace{-2pt}
\end{equation}
where $\omega := e^{j \frac{2\pi}{q}}$, $N = q^n$, $\mathbb{Z}_q$ denotes the ring of integers modulo $q$, and $\mathbb{Z}_q^n$ denotes the module of dimension $n$ over $\mathbb{Z}_q$. \hyperref[alg]{$q$-SFT} has a sample complexity of $O(Sn)$ and a computational complexity of $O(S n^2 \log q)$ in the noiseless case. In the presence of noise, for any fixed $q \geq 2$ \hyperref[alg]{$q$-SFT} has a sample complexity of $O(Sn^2)$ in $O(Sn^3)$ computations. In both cases, this is a significant improvement over the existing approaches.
We achieve this result via subsampling over affine spaces of the module $\Z_q^n$ and peeling-based erasure decoding \cite{Luby2001}.

Our approach to design a solution specific to $q$-ary sequences is a sharp departure from the common workaround in machine learning:
\emph{One-hot encoding} converts a $q$-ary sequence of length $n$ to a interpretable binary sequence of length $qn$ \cite{Gelman2021}. The downside of this approach is that the sequence length grows linearly in $q$, and thus the computational complexity of algorithms built on this approach is also at least linear in $q$. Furthermore, since the encoding is injective and non-surjective, there are many length $qn$ binary sequences that do not correspond to a $q$-ary sequence. For example, the all $0$s sequence is a valid $qn$ binary sequence but has no corresponding $q$-ary sequence. This \emph{sampling incompatibility} over linear spaces is fundamentally at odds with most sparse Fourier algorithms. Other approaches, such as binary coding of a $q$-ary sequence to one of length $\left \lceil \log_2 q\right \rceil n$, suffer a similar sampling incompatibility issue when $q$ is not a power of $2$, besides loosing the interpretability of one-hot encoding.

The potential wide range of applications for efficient $q$-ary Fourier transforms is a key motivation of this work. For example, in biological applications $F$ is used to identify which amino acids, or specific subsets of amino acids are responsible for certain functions \cite{aghazadeh2021epistatic, brookes2022sparsity}, but even in these cases computation remains a bottleneck.
In addition to this manuscript, we also release the associated software\footnote{\url{https://github.com/basics-lab/qsft}}, which will enable computation of $q$-ary Fourier transforms at scales never before seen.

\subsection{Contribution}
This paper studies the sparse Fourier transform of a $q$-ary function. We summarize the main contributions below:
\begin{compactitem}
\item We develop \hyperref[alg]{$q$-SFT}, the first algorithm for computing the Fourier transform of a $q$-ary function with sample complexity $O(S n)$ and computational complexity $(Sn^2 \log q)$ in the  presence of no additive noise. 
    \item  Noise robustness is critical for practical use, thus we also present a noise-robust version of $q$-SFT that for fixed $q > 1$ has sample complexity of $O(S n)$ at a cost of $O(nq^n)$ computational complexity. Alternatively, by increasing the sample complexity to $O(Sn^2)$, we show a computational complexity of $O(S n^3)$ can be achieved. 
\item Using systematic simulations on sparse $q$-ary functions, we show that \hyperref[alg]{$q$-SFT} performs well and is significantly more efficient than LASSO, allowing for computation at a massive scale.
    \item Numerical experiments also show \hyperref[alg]{$q$-SFT} performs well on real world data where the assumptions related to sparsity may not exactly hold.
\end{compactitem}

\subsection{Related Work}
At their core, sparse Fourier algorithms are built on the principles of sampling and aliasing in signal processing. In \cite{Li2015, Amrollahi2019}, which deal with the case of $q=2$, a pseudo-Boolean function is subsampled multiple times, according to an affine subsampling pattern. This leads to destructive aliasing. However, since the Fourier transform is sparse, and the affine subsampling method makes the aliasing predictable, a peeling decoding technique can be used to undo the aliasing and recover the original function.
Our algorithm in Section \ref{sec:noiseless} and \ref{sec:noisy} relies on similar techniques and shows that similar guarantees can be derived for the case of $q > 2$.
In \cite{Amrollahi2019} the case where $f$ has degree at most $t$ is also considered, and a sample complexity of $O(tS\log n)$ is achieved. 
In \cite{Ocal2019} the bounded degree problem is connected to perfect codes, however, the proposed algorithm is most efficient when Fourier coefficients are dense up to $t$th order.
A combinatorial approach to computing sparse Fourier transforms can be found in \cite{Choi2011}. 

There are other works that consider sparse Fourier transforms in other settings. Of note, \cite{Hassanieh2012, Pawar2018} both offer practical approaches for computing the sparse DFT of a signal. Furthermore, since the Fourier transform is a linear operation, computing the sparse Fourier transform may also be viewed as a linear inverse problem, and the existing approaches to \emph{compressed sensing} \cite{Donoho2006, candes2006stable} and \emph{group testing} \cite{Aldridge2019} are also relevant. In particular, LASSO \cite{tibshirani1996regression} and AMP \cite{donoho2009message, beck2009fast} can be used to solve for $F$ in \eqref{eq:q_fourier_transform}, but the computational complexity in both cases is exponential in $n$. 
We also note that decoding cyclic Reed-Solomon codes \cite{Blahut1983} is related to sparse Fourier transforms over finite fields.

\section{Problem Statement}

This work addresses the problem of computing the $S$-sparse Fourier transform $F$ of a $q$-ary function $f$ as in \eqref{eq:q_fourier_transform}. We consider a oracle model that takes in an index $\mathbf{m} \in \Z_q^n$, and outputs a function evaluation:
\begin{equation}\label{eq:sampling_model}
    \textrm{Oracle} : \mathbf{m} \longrightarrow f[\mathbf{m}] + v,\;\; v \sim \mathcal{CN}(0, \sigma^2),
\end{equation}
where $v$ is sampled i.i.d. from a complex Gaussian distribution (i.e., noise) for each query. \emph{Sample Complexity} denotes the number of oracle queries. The following assumptions are used to establish our theoretical guarantees.
\begin{assumption}\label{ass:sparse}
Let $F : \mathbb{Z}_q^n \to \mathbb{C}$ be a Fourier transform with support $\mathcal{S} := \mathrm{supp}(F)$. To facilitate our analysis throughout this paper, we make the following assumptions:
\begin{compactenum}
    \item Each element in the support set $\mathcal{S}$ is chosen independently and uniformly at random from $\mathbb{Z}_q^n$.
    \item The sparsity $S = |\mathcal{S}| = O(N^\delta)$ is sub-linear in $N$ for some $0 < \delta < 1$.
\end{compactenum}
\end{assumption}
\begin{assumption}\label{ass:sparse2}
\hfill
\begin{compactenum}
    \item Each coefficient $F[{\bf k}]$ for ${\bf k} \in \mathcal{S}$ is chosen from a finite set $\mathcal{X} := \{\rho, \rho \phi, \rho \phi^2, \dots, \rho \phi^{\kappa - 1} \}$ uniformly at random, where $\phi = e^{j 2 \pi / \kappa}$ and $\kappa$ is a constant.
    \item The signal-to-noise ratio (SNR) is defined as
    \vspace{-2pt}
    \begin{equation}
        \mathrm{SNR} = \frac{\|f\|^2}{N \sigma^2} = \frac{\|F\|^2}{ \sigma^2} = \frac{S \rho^2}{\sigma^2}
    \end{equation}
    \vspace{-2pt}
    and is assumed to be an arbitrary constant value, i.e. $\rho$ scales with $\sqrt{1/S}$.    
\end{compactenum}
\end{assumption}
\vspace{-6pt} \section{Algorithm: Noiseless}\label{sec:noiseless}
We begin by treating the case where samples of $f$ are obtained as in \eqref{eq:sampling_model} with $\sigma^2 = 0$. Algorithm~\ref{alg} describes the complete procedure, divided into three parts. 
In the first part, the function is subsampled, and the transform of the subsampled signal is computed.
The key innovation to deal with $q>2$ is to consider subsampling over an affine space in the module $\Z^n_q$. We show that if we subsample in this way the aliasing patterns are described by an affine function over the module.
In the second part, the subsampled transforms are processed, and singletons are identified by using the complex-valued transform coefficients to create a linear equation for a nonzero coefficient $\mathbf{k}$, again over the module. In 
the final phase, a peeling decoding process enables us to create more singletons, by subtracting the value of the originally identified singletons. Under assumptions, Theorem \ref{thm:noiseless} provides convergence guarantees. 
\begin{algorithm}
\caption{$q$-SFT}\label{alg}
\begin{algorithmic}[1]
\Require $P$, $C$, $b < n$, $\mathbf{M}_c \in \Z_q^{n \times b}$, $\mathbf{D}_c \in \Z_q^{P \times n}$ for $c \in [C]$.
\State $\hat{F}[\VEC{k}] \gets 0\;\; \forall \VEC{k}$
\For{$c \in [C]$} \Comment{Sub-Sampling Phase}
    \For{$p \in [P]$}
        \State Compute $U_{c,p}[\mathbf{j}]\;\;\forall \VEC{j} \in \Z_q^b$ \EndFor
    \State Compute $\type{\VEC{U}_{c}[\VEC{j}]} \;\;\forall \VEC{j} \in \Z_q^b$ \EndFor 
\State $\mathcal{K} = \emptyset$; $\mathcal{S} = \left\{(c, \VEC{j}, \VEC{k}, v) : \type{\VEC{U}_{c}[\VEC{j}]} = \mathcal{H}_{S}(\VEC{k}, v)\right\}$ \While{$\abs{\mathcal{S}} > 0$} \Comment{Peeling Phase}
\For{$(c, \VEC{j}, \VEC{k}, v) \in \mathcal{S}$, with $\VEC{k} \notin \mathcal{K}$}
    \State $\hat{F}[\VEC{k}] \gets v$;  $\mathcal{K} \gets \mathcal{K}\cup\{\VEC{k}\}$  
\For{$c' \in [C]$}
     \State ${\VEC{U}_{c'}[\VEC{M}_{c'}^\textrm{T} \VEC{k}]} \gets {\VEC{U}_{c'}[\VEC{M}_{c'}^\textrm{T} \VEC{k}]} - \hat{F}[\VEC{k}]\omega^{ \VEC{D}_{c'} \VEC{k} }$
     \State Recompute $\type{\VEC{U}_{c'}[\VEC{j}]}$
\EndFor
\EndFor
\State Update $\mathcal{S}$
\EndWhile
\Ensure Fourier Transform Estimate $\hat{F}$
\end{algorithmic}
\end{algorithm}
\subsection{Subsampling and Aliasing}
We construct $C$ subsampling groups, each characterized by a subsampling matrix $\VEC{M}_c \in \mathbb{Z}_q^{n \times b}$ defining a linear space and a set of $P$ offsets $\VEC{d}_{c,p} \in \mathbb{Z}_q^{n}$ for $p \in [P]$ for a total of $PB$ oracle queries with $B := q^b$. For each $c,p$ we compute:
\begin{equation}\label{eq:subsampled_fourier}
    U_{c,p}[\VEC{j}] = \frac{1}{B} \sum_{\boldsymbol\ell \in \mathbb{Z}_q^b} f[\VEC{M}_c \boldsymbol\ell + \VEC{d}_{c,p}] \omega^{-\langle \VEC{j}, \boldsymbol\ell \rangle},
\end{equation}
for all $\VEC{j} \in \mathbb{Z}_q^{b}$.
 In Appendix \ref{apdx:aliasing_proof}, we show that
\begin{equation}\label{eq:subsambpled_binned}
    U_{c,p}[\VEC{j}] = \sum_{\VEC{k}\; :\; \VEC{M}_c^\textrm{T} \VEC{k} = \VEC{j}} F[\VEC{k}] \omega^{\langle \VEC{d}_{c,p}, \VEC{k} \rangle},
\end{equation}
and thus the coefficients $U_{c,p}[\VEC{j}]$ are aliased versions of the coefficients $F[\VEC{k}]$.
It can be observed that the aliasing pattern is invariant with respect to the offsets $\VEC{d}_{c,p}$ used in subsampling. Therefore, we can group the observations according to their aliasing pattern to write $\VEC{U}_{c}[\VEC{j}] = [\dots, U_{c,p}[\VEC{j}], \dots]^{\textrm{T}}$
by stacking the $\VEC{j}$th coefficient associated with all the offsets. Then the aliasing pattern for the coefficients can be written as
\begin{equation}\label{eq:aliasing_pattern}
    \VEC{U}_{c}[\VEC{j}] = \sum_{\VEC{k} \;:\; \VEC{M}_c^\textrm{T} \VEC{k} = \VEC{j}} F[\VEC{k}] \omega^{ \VEC{D}_{c} \VEC{k} },
\ifdefined\ARXIV
\else
    \vspace{-5pt}
\fi
\end{equation}
where $\omega^{(\cdot)}$ is the element-wise exponentiation operator.

\subsection{Bin Detection}\label{subsec:bin_detect}

Each observation vector $\VEC{U}_{c}[\VEC{j}]$ in \eqref{eq:aliasing_pattern} is a linear combination of the unknown coefficients. The goal of bin detection is to identify which $\VEC{U}_{c}[\VEC{j}]$ correspond to singleton bins, and for what value $\VEC{k}$ that singleton corresponds to. To do so, define 
\begin{enumerate}
    \item $\type{\VEC{U}_{c}[\VEC{j}]} = \mathcal{H}_{Z}$ denotes a \emph{zero-ton}, for which there does not exist $F[\VEC{k}] \neq 0$ such that $ \VEC{M}_c^\textrm{T} \VEC{k} = \VEC{j}$.
    \item $\type{\VEC{U}_{c}[\VEC{j}]} = \mathcal{H}_{S}(\VEC{k}, F[\VEC{k}])$ denotes a \emph{singleton} with only one $\VEC{k}$ with $F[\VEC{k}] \neq 0$ such that $\VEC{M}_c^\textrm{T} \VEC{k} = \VEC{j}$.
    \item $\type{\VEC{U}_{c}[\VEC{j}]} = \mathcal{H}_{M}$ denotes a \emph{multi-ton} for which there exists more than one $F[\VEC{k}] \neq 0$ such that $ \VEC{M}_c^\textrm{T} \VEC{k} = \VEC{j}$.
\end{enumerate}
We use the following rule to compute  the bin type:
\begin{equation}\label{eq:noiseless_type}
\def\stackalignment{l}
   \type{\VEC{U}_{c}[\VEC{j}]} =\begin{cases}
        \mathcal{H}_{Z}, & \hspace{-2pt}\VEC{U}_{c}[\VEC{j}] = \VEC{0} \\
        \mathcal{H}_{M}, &\hspace{-2pt} \exists p \in [P] :\abs{\frac{U_{c,p}[\VEC{j}]}{ U_{c,1}[\VEC{j}]}} \neq 1 \\
        \mathcal{H}_{S}(\VEC{k}, X[\VEC{k}]), &\hspace{-2pt}\text{else.}
    \end{cases}
\end{equation}
When $\VEC{U}_{c}[\VEC{j}]$ is a singleton, we must determine the values $(\VEC{k}, F[\VEC{k}])$. 
Note that each singleton bin satisfies $U_{c,p}[\VEC{j}] = F[\VEC{k}] \omega^{\langle \VEC{d}_{c,p}, \VEC{k} \rangle}$ and hence we need to get rid of $F[\VEC{k}]$ to identify $\VEC{k}$. To achieve this, in addition to the offset matrix $\VEC{D}_c$, we choose a fixed delay $\VEC{d}_{c,0} = \VEC{0}$ such that $U_{c,0}[\VEC{j}] = F[\VEC{k}]$. Using these observations, we can write the following set of linear equations in $\VEC{k}$:
\begin{equation}\label{eq:singleton_index_recovery}
    \begin{bmatrix}
    \text{arg}_q[U_{c,1}[\VEC{j}]/U_{c,0}[\VEC{j}]]\\
\vdots\\
    \text{arg}_q[U_{c,P}[\VEC{j}]/U_{c,0}[\VEC{j}]]\\
    \end{bmatrix} = 
    \VEC{D}_c \VEC{k}.
\end{equation}
where $\text{arg}_q : \mathbb{C} \to \mathbb{Z}_q$ is the $q$-quantization of the argument of a complex number defined as  
\begin{equation}
    \text{arg}_q(z) := \left\lfloor \frac{q}{2 \pi}\text{arg}(z e^{j\pi/q})\right\rfloor.
\end{equation}
If we choose $\VEC{D}_c = \VEC{I}_n$, then $\VEC{k}$ can be obtained directly from \eqref{eq:singleton_index_recovery} and the value of the coefficient can be obtained as $F[\VEC{k}] = U_{c,0}[\VEC{j}]$. If we have additional information about $\mathbf{k}$, we can further reduce the number of rows in $\VEC{D}_c$ required to recover $\VEC{k}$. In particular, if $f$ has degree no greater than $t$, i.e., $\norm{\VEC{k}}_0 > t \implies F[\VEC{k}] = 0$, we can prove the following result.
\begin{proposition}\label{prop:source_coding}
    For any prime $q$ and $\VEC{k} \in \Z_q^n$ such that $\norm{\VEC{k}}_0 \leq t$ there exists a $\VEC{D}_c \in \Z_{q}^{P \times n}$ with $P = 2t\left \lceil \log_q(n)\right \rceil$ such that $\VEC{k}$ can be exactly recovered from $\VEC{D}_c \VEC{k}$.
\end{proposition}
 As a final note, we point out that Proposition \ref{prop:source_coding} can be used to significantly reduce the constant factor in \cite{Amrollahi2019} when $q=2$.
\subsection{Peeling Decoder}

Once we have determined the bin detection procedure, the final step is the use of a peeling decoder. Algorithm \ref{alg} describes the procedure in detail. It is analogous to decoding a code over a bipartite graph, where the variable nodes are the non-zero $F[\VEC{k}]$ and the check nodes are $U_{c}[\VEC{j}]$, respectively. There is an edge between $\VEC{U}_c[\VEC{j}]$ and $F[\VEC{k}]$ only if $\VEC{M}_c^\textrm{T} \VEC{k} = \VEC{j}$. After sub-sampling, the singletons are identified, and their values are subtracted from the $\VEC{U}_c$ values that they are connected to. This can be thought of as ``peeling" an edge of the bipartite graph. With the careful choice of parameters, we can prove the following performance guarantee.

\begin{figure}[ht]
    \centering
    \vspace{-6pt}
    \includegraphics[width=0.9\columnwidth]{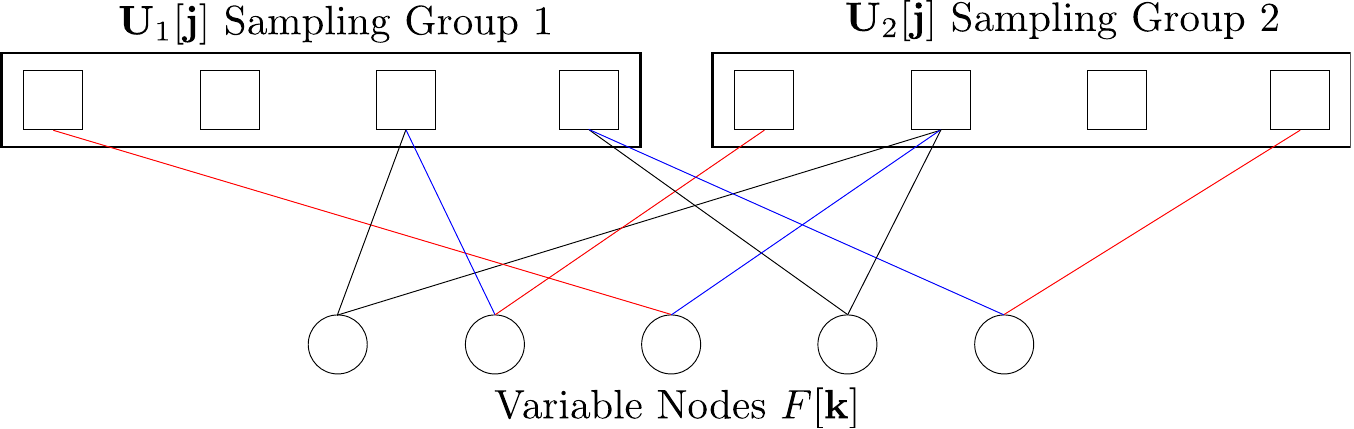}
    \caption{Bipartite graph representation of the peeling phase. Red edges correspond to singletons. These edges can be ``peeled'' from the graph and the value of the variable node that they are connected to can be determined. Then the blue edges can also be removed, revealing more singletons. This process iterates until either all edges have been peeled, or no singletons remain.}
    \label{fig:bipartite}
    \vspace{-12pt}
\end{figure}

\begin{theorem}[Noiseless Peeling Decoder Performance]\label{thm:noiseless}
    Consider a $q$-ary function $f$ as \eqref{eq:q_fourier_transform} satisfying Assumption \ref{ass:sparse}. For inputs $C = O(1)$, $P = n$, $q^b = O(S)$, $\VEC{D}_c = \mathbf{I}_n$, there exist some $\VEC{M}_c$ such that the output \hyperref[alg]{q-SFT} is exactly $F$ with probability at least $1 - O(1/S)$. With these inputs \hyperref[alg]{q-SFT} requires $O(Sn)$ samples and $O(S n^2 \log q)$ computations.
\begin{proof}  See Appendix \ref{apdx:peeling_thm}.
\end{proof}
\end{theorem}

 \section{Algorithm: Noise Robust}\label{sec:noisy}

The key to robustness is changing the bin detection scheme of Section~\ref{subsec:bin_detect} to account for noise by subsampling carefully.
We have seen that the offset signature $\omega^{ \VEC{D} \VEC{k}}$ is the key to decoding the unknown pair $(\VEC{k}, F[\VEC{k}])$. Denoting $\VEC{s}_{\VEC{k}} = \omega^{ \VEC{D} \VEC{k}}$,
let $\VEC{S} = [ \dotsc, \VEC{s}_{\VEC{k}}, \dotsc]$ be the offset signatures associated with offsets $\VEC{D}$. Then, in the presence of noise the bin observation vector $\VEC{U}$ \footnote{For
simplicity in this section, we drop the group index $c$ and bin index $\VEC{j}$ when we mention bin observations, e.g. we write $\VEC{U}$ to denote $\VEC{U}_{c}[\VEC{j}]$. } can be written as
$    \VEC{U} = \VEC{S} \VEC{\alpha} + \VEC{W}
\label{eqn:matrix_observations}$
for some sparse vector $\VEC{\alpha} = [\dots, \alpha[\VEC{k}], \dots]^{\mathrm{T}}$ such that $\alpha[\VEC{k}] = F[\VEC{k}]$ if $ \VEC{M}^\mathrm{T} \VEC{k} = \VEC{j}$ and $\alpha[\VEC{k}] = 0$ otherwise. In the case of single-tons, $\VEC{\alpha}$ is 1-sparse and therefore $\VEC{U}$ can be regarded as the noise-corrupted version of some code word from the codebook $\VEC{S}$. Further, it can be shown that $\VEC{W}$ has multivariate complex Gaussian distribution with zero mean and covariance $\nu^2 \VEC{I}$ where $\nu^2 := \sigma^2 / B$.
\subsection{Near-linear Time Robust Bin Detection}\label{subsec:robust_bin_detect}

The near-linear (in $N$) time bin identification scheme uses $P$ random offsets $\VEC{d}_p$ for $p \in [P]$ chosen independently and uniformly at random over $\mathbb{Z}_q^n$. For some $\gamma \in (0,1)$, detection is performed by proceeding with the following steps.
\begin{enumerate}[leftmargin=\parindent, align=left, labelwidth=\parindent, labelsep=1pt]
    \item \textbf{Zero-ton verification: } We first rule out zero-tons. We declare $\type{\VEC{U}} = \mathcal{H}_{Z}$ if
$
    \frac{1}{P} \|\VEC{U}\|^2 \leq (1+\gamma) \nu^2.
$
    \item \textbf{Single-ton search: } The next step is to estimate $(\VEC{k}, F[\VEC{k}])$ assuming that bin $\VEC{j}$ is a singleton. We compute the Maximum Likelihood Estimate (MLE) of $\VEC{k}$, which involves a search over the
$N/B = q^{n-b}$ possible values of $\VEC{k}$ that satisfy $\VEC{M}_c^\textrm{T} \VEC{k} = \VEC{j}$.
    For each location $\VEC{k}$, we write MLE of the single-ton coefficient as
        $\widehat{\alpha}[\VEC{k}] = \VEC{s}_{\VEC{k}}^\textrm{T} \VEC{U} / P$, thus
    \begin{equation}
        \widehat{\VEC{k}} = \argmin_{\VEC{k} \;:\; \VEC{M}_c^\textrm{T} \VEC{k} = \VEC{j}} \|\VEC{U} - \widehat{\alpha}[\VEC{k}] \VEC{s}_{\VEC{k}}\|^2.
    \end{equation}
    Given that $F$ satisfies Assumption~\ref{ass:sparse2}, we estimate the value of the coefficient corresponding to $\widehat{\VEC{k}}$ as
    \begin{equation}
     \widehat{F}[\widehat{\VEC{k}}] = \argmin_{\alpha \in \mathcal{X}} \| \alpha - \VEC{s}_{\widehat{\VEC{k}}}^\textrm{T} \VEC{U} / P\|.
     \label{eqn:value_estimation}
    \end{equation}
    \item \textbf{Singleton verification: } This step confirms whether  the bin with estimated singleton pair $(\widehat{\VEC{k}}, \widehat{F}[\widehat{\VEC{k}}])$ is a singleton via a residual test. We declare $\type{\VEC{U}} = \mathcal{H}_{S}(\widehat{\VEC{k}}, \widehat{F}[\widehat{\VEC{k}}]) $ if
\begin{equation}
    \frac{1}{P} \norm{\VEC{U} - \widehat{F}[\widehat{\VEC{k}}] \VEC{s}_{\widehat{\VEC{k}}} }^2 \leq (1+\gamma) \nu^2.
\end{equation}
    Otherwise, we declare $\type{\VEC{U}} = \mathcal{H}_{M}$.
\end{enumerate}

\subsection{Sub-linear Time Robust Bin Detection}

The design described in the previous section requires an exhaustive search due to a lack of structure in random offsets. In order to overcome this bottleneck in computational complexity, we design offsets that enable symbol-by-symbol recovery of the singleton index $\VEC{k}$. We let $P = P_1 n$ and generate $P_1$ random offsets $\VEC{d}_p$ for $p \in [P_1]$ chosen independently and uniformly at random over $\mathbb{Z}_q^n$. We continue to perform zero-ton and singleton verification steps using this set of $P_1$ random offsets as in the previous section. However, to achieve sub-linear time for the single-ton search step, we generate $n$ modulated offsets $\VEC{d}_{p,r}$ for each $p \in [P_1]$ such that 
\begin{equation} \label{eq:nso_delay}
    \VEC{d}_{p} \oplus_q \VEC{e}_{r} = \VEC{d}_{p,r}, \quad \forall r \in [n]
\end{equation}
where $\VEC{e}_{r}$ is the $r$-th column of the identity matrix. Given these offsets, we can identify the $r$-th symbol of $\VEC{k}$ by jointly considering the observations associated with offsets $\{\VEC{d}_{p,r}\}_{p \in P_1}$. In particular, the observations satisfy the following proposition.
\begin{proposition}\label{prop:1}
Given a single-ton bin $(\VEC{k}, F[\VEC{k}])$, the $q$-quantized argument of the ratio of observations
\begin{equation}
    U_{p} = F[\VEC{k}] \omega^{\langle \VEC{d}_{p}, \VEC{k} \rangle} + W_{p},\;\;
    U_{p,r} = F[\VEC{k}] \omega^{\langle \VEC{d}_{p,r}, \VEC{k} \rangle} + W_{p,r}
\end{equation}
satisfies
    $\text{arg}_q[U_{p,r} / U_{p}] = \langle \VEC{e}_{r}, \VEC{k} \rangle \oplus_q Z_{p, r}$
where $Z_{p, r}$ is a random variable over $\mathbb{Z}_q$ with $p_i := \mathbb{P}(Z_{p,r} = i)$. The distribution satisfies $p_i < p_0$ for all $i \neq 0$ and $\sum_{i \neq 0} p_i \leq \mathbb{P}_e := 2 e^{-\frac{\zeta}{2} \text{SNR}}$ for $\zeta := \eta \sin^2{(\pi/2q)}$. 
\label{prop:quant}
\end{proposition}

\begin{proof}
    See Appendix \ref{proof:prop_quant}
\end{proof}

Based on this observation, we can apply a majority test to estimate $r$-th entry of $\VEC{k}$ as
\begin{equation}
    \hat{k}[r] = \arg \max_{a \in \mathbb{Z}_q} \sum_{p \in [P_1]} \mathbb{1}\{a = \text{arg}_q[U_{p,r} / U_{p}]\}
\end{equation} 

Then, using the estimation $\widehat{\VEC{k}}$ of the index, the estimation for the value of the coefficient is obtained as in \eqref{eqn:value_estimation}.

\begin{theorem}[Robust Peeling Decoder Performance]\label{thm:noisy}
    Consider a $q$-ary function $f$ as \eqref{eq:q_fourier_transform} satisfying Assumption \ref{ass:sparse} and \ref{ass:sparse2}. For inputs $C = O(1)$, $P = n$, $q^b = O(S)$, $\VEC{D}_c$ chosen uniformly at random over $\Z_q^{P \times n}$, there exist some $\VEC{M}_c$ such that the output \hyperref[alg]{q-SFT} is exactly $F$ with probability at least $1 - O(1/S)$. With these inputs \hyperref[alg]{q-SFT} requires $O(S n)$ samples and $O(nq^n)$ computations as $n$ grows with fixed $q$. If instead, we have $P = n^2$ and $\VEC{D}_c$ chosen as in \eqref{eq:nso_delay}, the same result holds with $O(Sn^2)$ samples and $O(S n^3)$ computations as $n$ grows.
\end{theorem}
\begin{proof}
    See Appendix \ref{apdx:peeling_thm_noisy}. The ideas are similar to Theorem \ref{thm:noiseless}, but it requires the use of Proposition \ref{prop:1} to deal with noise.
\end{proof} 

\section{Numerical Experiments}\label{sec:results}
In this section, we provide numerical experiments that showcases the performance of \hyperref[alg]{$q$-SFT} in practice. In addition to studying the algorithm's performance in recovering synthetically generated signals that are sparse in the Fourier transform domain, we discuss the problem of recovering the mean free energy values of RNA given its sequence.
We consider the following simulation setting. For given  parameters $q$ and $n$, we synthetically generate a signal $f$ such that its $q$-ary Fourier transform $F$ is $S$-sparse and $\mathcal{S} = \mathrm{supp}(F)$
chosen uniformly at random over $\mathbb{Z}_q^n$ with values \begin{equation}
F[\mathbf{k}] = 
    \begin{cases}
    V[\mathbf{k}] e^{-j \Omega [\mathbf{k}]}, & \text{if } \mathbf{k} \in \mathcal{S}\\
    0 &\text{otherwise}.
    \end{cases}
\end{equation}
where $V[\mathbf{k}] \in [\rho_{\mathrm{min}}, \rho_{\mathrm{max}}]$ are independent random variables and $\Omega[\mathbf{k}]$ are independent uniform random variables over $[0, 2 \pi)$. To demonstrate the robust recovery performance of $q$-SFT in the presence of noise, we assume access to noisy observations as in \eqref{eq:sampling_model}. 
Note that our theoretical results for robust recovery are given under Assumption \ref{ass:sparse2} which corresponds to having $\rho_{\mathrm{min}} = \rho_{\mathrm{max}} = \rho$ and $\Omega[\mathbf{k}]$ taking values over a finite set of size $\kappa$, however, we consider a more general scenario 
to demonstrate that our algorithm can achieve robust recovery in more general settings.

We measure the accuracy of the outputs using the normalized mean-squared error (NMSE) metric defined as
\begin{equation}
    \mathrm{NMSE} = \frac{\|\widehat{F} - F\|^2} {\|F\|^2} =  \frac{\|\widehat{f} - f\|^2} {\|f\|^2}.
\end{equation}

Fig. ~\ref{fig:a} and ~\ref{fig:b} compare $q$-SFT with LASSO in terms of the number of samples and runtime. To implement LASSO in this complex-valued problem, we consider group regularization \cite{yuan2006model} followed by ridge regression refinement \cite{hoerl1970ridge}. We set 
$q = 3$, $\rho_{\mathrm{min}} = 1$, $\rho_{\mathrm{max}} = 5$, $S = 100$ and $\sigma^2$ such that SNR is fixed at 10dB.
For $q$-SFT, we vary hyper-parameters $b$, $C$, and $P$ such that each combination results in a different number of samples used. Similarly, we set the number of uniformly chosen samples over $\mathbb{Z}_q^n$ provided to LASSO to compare its performance with $q$-SFT at a given sample complexity. LASSO is successful in recovering Fourier coefficients for small problems, but Fig.~\ref{fig:b} indicates that $q$-SFT's runtime scales linearly in $n$ whereas LASSO's scales exponentially in $n$ (linearly in $N$). 
While our algorithm continues to work in a reasonable time for problem sizes as large as $n=18$ in under $0.1$ seconds, LASSO does not run on our computer for larger problems. 

Fig.~\ref{fig:c} shows the relation between noise level and recovery performance. We fix $n= 20$, $q=4$ and $S$ as well as hyper-parameters $b$, $C$ and $P$. We run $q$-SFT at various SNRs, observing a sharp phase transition where our algorithm is successful at high SNR and fails at low SNR. As expected, the transition point is at a lower SNR for sparser signals.

Finally, Fig.~\ref{fig:d} depicts the performance of $q$-SFT on learning a real-world function $f$. Specifically, we consider ViennaRNA \cite{Lorenz2011}, a computational tool that computes the Mean Free Energy (MFE) of RNA sequences. We represent each base with an element of $\mathbb{Z}_4$ and denote different length-$n$ RNA sequences by $\mathbf{k} \in \mathbb{Z}_4^n$. Then, we consider $\textrm{MFE}(\mathbf{k})$ computed by ViennaRNA to be noisy observations of $f[\mathbf{k}]$. We run $q$-SFT with these noisy observations to construct $\widehat{F}$ (hence $\widehat{f}$) and calculate a Test NMSE by calculating the NMSE only over a uniformly randomly chosen support. As can be seen in Fig.~\ref{fig:d}, in all cases $q$-SFT is able to achieve a Test NMSE of $0.1$ when enough samples are provided. For instance, when $n=18$, only $0.01\%$ of the total number of sequences is sufficient to achieve a test NMSE of less than $0.1$.  
 
\begin{figure*}[ht]
\ifdefined\ARXIV
     \centering
     \begin{subfigure}[h]{0.48\textwidth}
         \centering
         \caption{\label{fig:a}}
         \includegraphics[width=\textwidth]{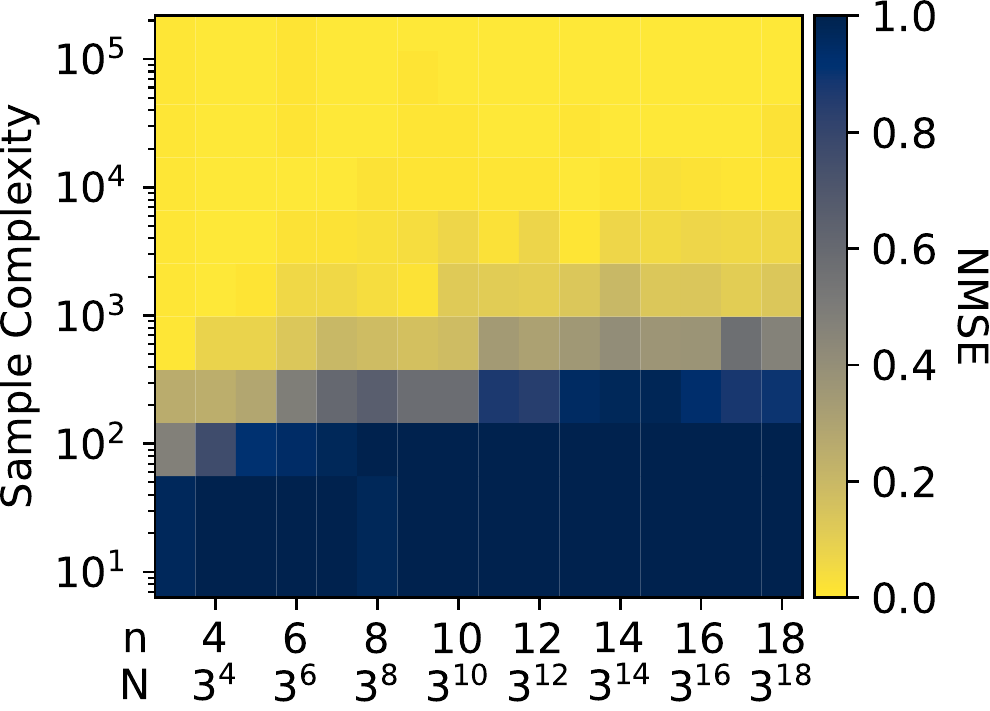}
   \end{subfigure}
   \hfill
     \begin{subfigure}[h]{0.48\textwidth}
         \centering
         \caption{\label{fig:b}}
         \includegraphics[width=\textwidth]{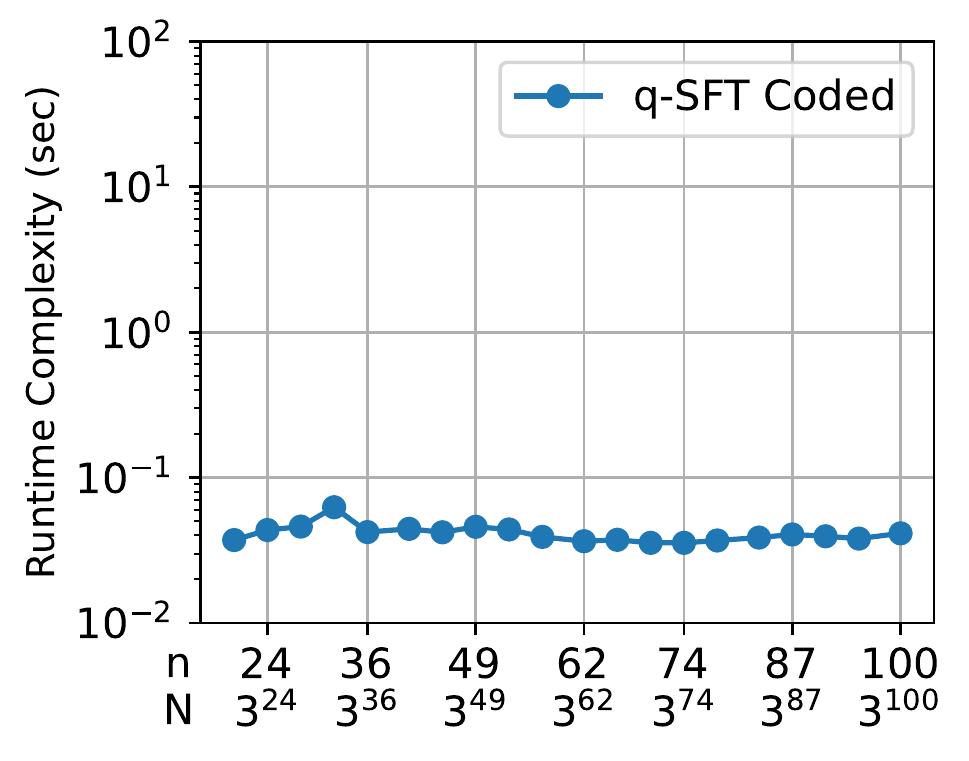}
     \end{subfigure}
     \begin{subfigure}[h]{0.48\textwidth}
         \centering
         \caption{\label{fig:c}}
         \includegraphics[width=\textwidth]{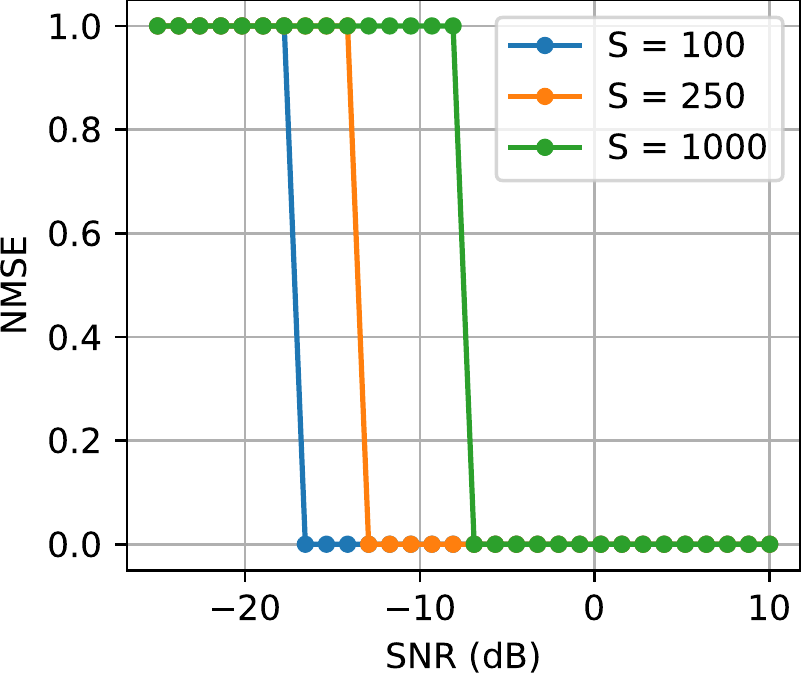}
     \end{subfigure}
    \hfill
     \begin{subfigure}[h]{0.48\textwidth}
         \centering
         \caption{\label{fig:d}}
         \includegraphics[width=\textwidth]{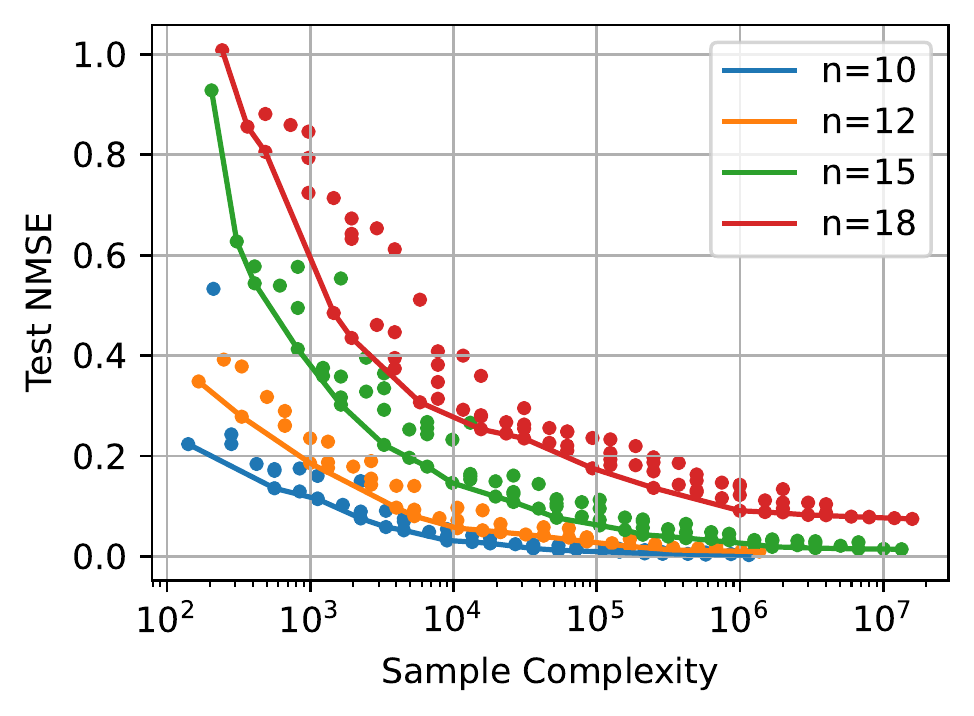}
     \end{subfigure}
\else
\vspace{-10pt}
     \centering
     \begin{subfigure}[h]{0.24\textwidth}
         \centering
         \caption{\label{fig:a}}
         \includegraphics[width=\textwidth]{complexity-vs-n-qspright.pdf}
   \end{subfigure}
   \hfill
     \begin{subfigure}[h]{0.24\textwidth}
         \centering
         \caption{\label{fig:b}}
         \includegraphics[width=\textwidth]{complexity-vs-n-runtime.pdf}
     \end{subfigure}
     \begin{subfigure}[h]{0.24\textwidth}
         \centering
         \caption{\label{fig:c}}
         \includegraphics[width=\textwidth]{nmse-vs-snr.pdf}
     \end{subfigure}
    \hfill
     \begin{subfigure}[h]{0.24\textwidth}
         \centering
         \caption{\label{fig:d}}
         \includegraphics[width=\textwidth]{complexity-vs-n-rna.pdf}
     \end{subfigure}
\fi
    \caption{
    \textbf{(a)} NMSE of $q$-SFT for a range of sample complexity and $N$ on synthetically generated data (note that the transition threshold appears logarithmic in $N$), \textbf{(b)} runtime plotted against $N$ for the experiment in (a) as well as LASSO on the same data (note the exponential runtime of LASSO, while $q$-SFT is sub-linear in $N$), \textbf{(c)} NMSE phase transition of $q$-SFT against SNR (note that $q$-SFT is successful for lower SNR if the sparsity is lower),
    \textbf{(d)} performance of $q$-SFT on the Mean Free Energy function of RNA sequences ($q=4$) as computed by the ViennaRNA \cite{Lorenz2011} (for the larger values of $n$ pictured, $q$-SFT needs only to query a small fraction of the total function evaluations to achieve a test NMSE of less than $0.1$).}
\vspace{-10pt}
\end{figure*}

\section{Conclusion and Future Work}
This manuscript presents a fast and sample efficient algorithm to computing $S$ sparse Fourier transforms of $q$-ary functions. With the wide range of problems that are well modeled by $q$-ary functions, our algorithm $q$-SFT can be applied broadly. We identify recent advancements in deep generative modeling enabled by AlphaFold \cite{Jumper2021} and protein language models \cite{Lin2022} as a significant potential area of application; they allow us to sample biological functions at scale, enabling a compact and explainable representation in terms of high-order sparse polynomials using $q$-ary Fourier transform. Furthermore, though the progress we make in this work is significant, challenges still remain. For example, in the aforementioned biological application $q$-ary functions of bounded degree are of interest. We speculate that that for prime $q$, BCH codes can be used to design efficient offsets $\mathbf{D}$, and it may be possible to design more efficient algorithms in general in that case.

\newpage
\appendix
\subsection{Future Work: Bounded Degree}\label{apdx:low-weight}

In this section we prove Proposition \ref{prop:source_coding}. We note that when the Hamming weight of $\VEC{k}$ is low, \eqref{eq:singleton_index_recovery} can be viewed as a set of parity check equations, and $\VEC{D}_c\VEC{k}$ the corresponding syndrome. Thus, if $\VEC{D}_c$ is a parity check matrix for a $t$-error correcting code, $\VEC{k}$ can be reconstructed. 
\begin{proof}
    Note that since $q$ is prime, $\Z_q$ is equivalent to the field $\mathbb{F}_q$. Let $c \triangleq \left \lceil \log_q(n)\right \rceil$ for brevity. There exists the $t$-error correcting Reed-Solomon code $\mathcal{C}_{\mathrm{RS}} = \mathrm{RS}(q^c, q^c - 2t)_{q^c}$. Thus, there exists a sub-field-sub-code $\mathcal{C}_{\mathrm{BCH}} = \mathbb{F}^{q^c}_q \cap \mathcal{C}_{\mathrm{RS}}$. Furthermore, such a code can be shortened by $q^c - n$ symbols to obtain $\mathcal{C}'_{\mathrm{BCH}}$. Since shortening and considering sub-codes does not reduce the minimum distance, the shortened code is still $t$-error correcting. Let $\VEC{H}_{\mathrm{BCH}} \in \mathbb{F}_q^{2tc \times n}$ represent the parity check matrix of $\mathcal{C}'_{\mathrm{BCH}}$ over $\mathbb{F}_q$. Taking $\VEC{D}_c = \VEC{H}_{\mathrm{BCH}}$ gives us the desired result.
\end{proof}

Ideas from Section \ref{sec:noisy} can be applied on top of Proposition \ref{prop:source_coding} for robustness. In practice, we find that $q$-SFT with $\VEC{D}_c$ chosen as above performs well with random $\VEC{M}_c$ in the bounded degree problem. Due to lack of independence in the aliasing pattern however, a rigorous proof has remained elusive. We also point out that Proposition \ref{prop:source_coding} can be used to significantly reduce the constant factor in \cite{Amrollahi2019}.

Note that Proposition \ref{prop:source_coding} holds only for prime $q$ because we leverage algebraic coding theory, which has mostly focused on the problem of defining codes over fields. While it is possible to construct BCH codes over rings \cite{Blake1972, Shankar1979}, they are generally inferior and yield a weaker result. 

\ifdefined\ARXIV

Fig.~\ref{fig:a2} and ~\ref{fig:b2} show the performance of coded $q$-SFT in terms of the number of samples and runtime. We set 
$q = 3$, $\rho_{\mathrm{min}} = 1$, $\rho_{\mathrm{max}} = 1$, $S = 1000$ $t = 5$, and $\sigma^2$ such that SNR is fixed at 20dB. We vary hyper-parameters $b$, $C$, and $P$ such that each combination results in a different number of samples used. 

\begin{figure*}[h]
     \begin{subfigure}[h]{0.48\textwidth}
         \centering
         \caption{\label{fig:a2}}
         \includegraphics[width=\textwidth]{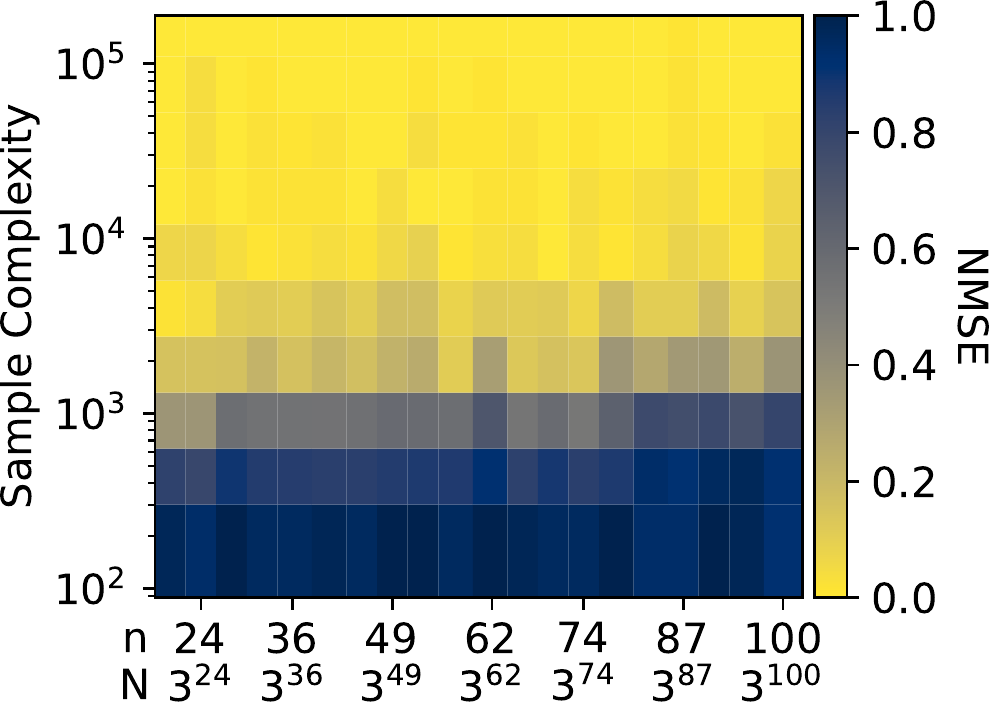}
     \end{subfigure}
    \hfill
     \begin{subfigure}[h]{0.48\textwidth}
         \centering
         \caption{\label{fig:b2}}
         \includegraphics[width=\textwidth]{complexity-vs-n-runtime.pdf}
     \end{subfigure}
    \caption{
    \textbf{(a)} NMSE of coded $q$-SFT for a range of sample complexity and $N$ values on synthetically generated data (note that the transition threshold appears sub-logarithmic in $N$),
    \textbf{(b)} runtime plotted against $N$ for the experiment in (a)}
\end{figure*}

\fi 
\subsection{Proof of the Aliasing Pattern}\label{apdx:aliasing_proof}

First, we assume $\mathbf{d}_p = \VEC{0}$. 
\begin{IEEEeqnarray}{rCl}
    U_{c,p}[\mathbf{j}] & = & 
    \frac{1}{B} \sum_{ \boldsymbol{\ell} \in \Z_q^b}  f[\mathbf{\mathbf{M}  \boldsymbol{\ell}}] \omega^{-\langle \mathbf{j} , \boldsymbol{\ell}\rangle}, \\
    & = &
        \frac{1}{B}\sum_{ \VEC{\ell} \in \Z_q^n} \left( \sum_{\mathbf{k} \in \Z_q^n} F[\mathbf{k}] \omega^{\langle \mathbf{M} \VEC{\ell} , \mathbf{k}\rangle} \right) \omega^{-\langle \mathbf{j} , \boldsymbol{\ell}\rangle}, \\
    & = &
        \frac{1}{B}\sum_{\mathbf{k} \in \Z_q^n} F[\mathbf{k}] \left( \sum_{\VEC{\ell} \in \Z_q^b}  \omega^{\langle \mathbf{M}^{\mathrm{T}}\mathbf{k} - \mathbf{j}, \VEC{\ell}\rangle} \right), \\
    & = & 
    \sum_{\mathbf{k}\;:\;\mathbf{M}^{\top} \mathbf{k} = \mathbf{j}} F[\mathbf{k}].
\end{IEEEeqnarray}
where the last inequality follows from Lemma \ref{lemma:zero}. Finally, by using the shifting property, for arbitrary $\mathbf{d}_{p}$, we have:
\begin{equation}
    U_{c,p}[\mathbf{j}] = \sum_{\mathbf{k}\;:\;\mathbf{M}^{\mathrm{T}} \mathbf{k} = \mathbf{j}} F[\mathbf{k}] \omega^{\langle \mathbf{d}_p, \mathbf{k} \rangle }. 
\end{equation}

\begin{lemma}
Let $\omega = \exp(-i 2\pi/q)$, and $\boldsymbol{a} \in \Z_q^b$. Then,\\
\begin{equation}
    0 = \sum_{\boldsymbol{\ell} \in \mathbb{F}_q^b} \omega^{\langle \boldsymbol{\ell} , \boldsymbol{a}\rangle} \iff \boldsymbol{a} \neq \boldsymbol{0}
\end{equation}
\label{lemma:zero}
\end{lemma}

\begin{proof}
\begin{IEEEeqnarray}{rCl}
    \sum_{\boldsymbol{\ell} \in \Z_q^b} \omega^{\langle \boldsymbol{\ell} , \boldsymbol{a}\rangle} &=& \sum_{\boldsymbol{\ell} \in \Z_q^b} \omega^{\sum_{i=1}^{b} \ell_i a_i} \\
    &=& \sum_{\boldsymbol{\ell} \in \Z_q^b} \prod_{i=1}^{b} \omega^{\ell_i a_i} = \prod_{i=1}^{b}\sum_{\ell_i \in \Z_q} \omega^{\ell_i a_i}
\end{IEEEeqnarray}

For proving the forward direction, if $\boldsymbol{a} = \boldsymbol{0}$, then the right hand is equal to $B = q^b$ which is nonzero.
For proving the converse, if $\boldsymbol{a} \neq \boldsymbol{0}$, then at least for one $i$, we have $a_i \neq 0$. Since $\omega$ is the $q$-th root of unity, the $i$-th term of the product is equal to $\sum_{\ell_i \in \mathbb{F}_q} \omega^{\ell_i a_i} = 0$ making the right hand side equal to 0.
\end{proof}

\subsection{Proof of Theorem \ref{thm:noiseless}}
\label{apdx:peeling_thm}
We let $\mathcal{G}(S, \eta, C, \{\VEC{M}_c\}_{c \in [C]})$ represent the set of all bipartite graphs that are induced by subsampling with $B = \eta S$ and subsampling matrices $\{\VEC{M}_c\}_{c \in [C]}$. By Assumption \ref{ass:sparse}, the graph that we must decode over in Algorithm \ref{alg} is uniformly distributed over this set. 

By the choice of $\VEC{D}_c$, for any singleton index $U_c[\mathbf{j}]$, $F[\mathbf{k}]$ and $\mathbf{k}$ are retrieved with probability $1$ using the methods described in Section \ref{subsec:bin_detect}. What remains then is to show that taking $C$ to be $O(1)$ is sufficient for the peeling to peel every edge in the graph. Due to our uniformly distributed support assumption, the argument is essentially identical to the binary case in \cite{Li2015} Appendix B. In this section, we sketch the proof for $0 \leq \delta \leq 1/3$, but the other cases follow a similar argument, and are described in \cite{Li2015}.
In this case, $\mathbf{M}_c$ may be constructed as:
\begin{equation}\label{eq:degree_dist}
    \VEC{M}_c = [\VEC{0}_{b \times (c-1)}, \VEC{I}_{b \times b}, \VEC{0}_{b \times (n-cb)}]^{\mathrm{T}}, \;\; c\in[C]
\end{equation}
Choosing $\VEC{M}_c$ in this way ensures that the edge degree distribution is easily computed. Specifically, the number of edges connected to check nodes with degree $j$ when $\eta S = B$ is:
\begin{equation}
    \rho_j = \frac{({1}/{\eta})^{j-1}e^{-1/\eta}}{(j-1)!} ,\; j=0,\dotsc,S,
\end{equation}
and zero otherwise. For the next step, we define the neighborhood of an edge $e = (v, c)$, where $v$ and $c$ are the connected variable and check node respectively as follows: $\mathcal{N}_{e}^{\ell}$ is the induced sub-graph containing the edges and nodes of all paths $e_1,\dotsc, e_{\ell}$ where $v \in e_1$ and $e_1 \neq e$. If we consider an arbitrary edge in our decoding graph, such that $\mathcal{N}_e^{2i}$ is a tree, the probability that it will not be removed after the $i$th iteration $p_i$ can be written as: $p_i = \left( 1 - \sum_j \rho_j \left(1 - p_{i-1}\right)^{j-1}\right)^{C-1}$, for sufficiently large $S$ this is well approximated by
\begin{equation}
    p_i = \left(1 - e^{-\frac{1}{\eta}p_{i-1}}\right)^{C-1}.
\end{equation}
From this equation, we can see that for any choice of $C$ there exists an $\eta$ such that $p_i$ goes to zero as $i$ does. 
 Then let $\mathcal{T}_i$ denote the event that for every edge $e$ the neighbourhood $\mathcal{N}_e^{2i}$ is a tree, and let $Z_i$ represent the number of edges that are still not decoded after the $i$th peeling iteration. Let $Z_i = \sum_{e} Z_i^e$, where $Z_i^e = \mathds{1}\{\text{Edge}\; e\;\text{is not peeled}\}$.
We can now bound $\E\left[ Z_i\lvert \mathcal{T}_i\right]$ as follows:
\begin{equation}
    \E\left[ Z_i\lvert \mathcal{T}_i\right] = \sum_{e} \left[ Z_i\lvert \mathcal{T}_i\right] = CSp_i.
\end{equation}
Since it is possible to choose $\eta$ such that $p_i$ can be made arbitrarily small, we conclude that
for $\eps > 0$ there exist some $i$ such that $\E\left[ Z_i\lvert \mathcal{T}_i\right] = CS\eps/4$.
Then, via Lemma 6 in \cite{Li2015}, which is similar to \cite{Richardson2001}, we can show that for large $S$, $\mathcal{T}_i$ occurs with high probability. Thus we can bound the absolute difference between the conditional and unconditional mean as:
\begin{equation}
    \abs{\E\left[ Z_i\right] - \E\left[ Z_i\lvert \mathcal{T}_i\right]} \leq CS \eps/4,
\end{equation}
allowing us to conclude that there is some $i$ such that $\E\left[ Z_i\right] \leq CS\eps/2$ for $S$ greater than some constant. We can further show that $Z_i$ is well concentrated around its mean by constructing a suitable martingale, and applying Azuma's Inequality. This allows us to establish:
\begin{equation}\label{eq:th1_concentration}
    \mathrm{Pr}\left( \abs{Z_i - \E\left[ Z_i\right]} > CS\eps/2 \right) \leq 2\mathrm{exp}\left( -\beta \eps^2 S^{\frac{1}{4i+1}}\right),
\end{equation}
for some constant $\beta> 0$. Thus, we have shown that with high probability, our peeling phase terminates with the number of remaining edges $Z_i < CS\eps$ for any $\eps > 0$. Note now, if we could choose $\eps = 1/CS$, we would be done, however, this would cause our bound in \eqref{eq:th1_concentration} to be meaningless. Instead, we use graph expander properties to complete the argument. 

We call a bipartite graph in $\mathcal{G}(S, \eta, C, \{\VEC{M}_c\}_{c \in [C]})$ an $\eps$-expander if for all subsets $\mathcal{S}$ of variable nodes with $\abs{\mathcal{S}} < \eps S$, there exists a check node neighbourhood of $\mathcal{S}$ in one of the sub-sampling groups $c$ denoted as $\mathcal{N}_c(\mathcal{S})$ that satisfies $\abs{\mathcal{N}_c(\mathcal{S})} > \abs{S}/2$. It can be easily shown, based on \eqref{eq:degree_dist} that if the graph is uniformly chosen over $\mathcal{G}(S, \eta, C, \{\VEC{M}_c\}_{c \in [C]})$, it is an $\eps$-expander with probability at least $1 - O(1/S)$ if $C \geq 3$. This is done in Appendix B.7 of \cite{Li2015} using a counting argument and elementary inequalities.

Let the set of $Z_i$ remaining edges be connected to the variable nodes $\mathcal{S}$. A sufficient condition for all the right nodes in at least one group $\mathcal{N}_c(\mathcal{S})$ to have at least one singleton is that the corresponding average degree is less than 2, which implies that $\abs{S}/\abs{\mathcal{N}_c(\mathcal{S})} \leq 2$ and hence $\abs{\mathcal{N}_c(\mathcal{S})} \geq \abs{S}/2$. Since the graph is an expander with probability at least $1 - O(1/S)$ this condition is satisfied, and all the edges will be peeled with at least this probability. 

Note that the number of samples required is $O(PCB) = O(Sn)$. The total number of computations is dominated by the subsampling, which requires $O(PB\log B) = O(PS\log S) = O(PS\log N) $ because $S = O(N^\delta)$. As a result, this gives a total complexity of $O(PSn\log q) = O(S n^2 \log q)$. 
\qed
\subsection{Proof of Proposition \ref{prop:quant}}
\label{proof:prop_quant}

Given a singleton bin with an index-value pair $(\VEC{k}, F[\VEC{k}])$, we can write
\begin{equation*}
    \text{arg}[U_{p}] = \text{arg}[F[\VEC{k}]] + \frac{2 \pi}{q}\langle \VEC{d}_{p}, \VEC{k} \rangle + Y_{p}
\end{equation*}
where the additions are modulo-$2 \pi$ and $Y_{p}$ is a random variable over $[-\pi, \pi)$ that satisfies
\begin{align*}
    \mathrm{Pr}(|Y_{p}| \geq \alpha) &\leq \mathrm{Pr}(|W_{p}| \geq |F[\VEC{k}]| \sin{(\alpha)})\\
    &\leq \exp \left(-\frac{|F[\VEC{k}]|^2 \sin^2{(\alpha)}}{2\sigma^2/B} \right)\\
    &= \exp \left(-\frac{\eta \sin^2{(\alpha)}}{2} \text{SNR} \right)
\end{align*}
for $0 \leq \alpha \leq \pi/2$. Similarly, we can write
\begin{equation*}
    \text{arg}[U_{p,r}] = \text{arg}[F[\VEC{k}]] + \frac{2 \pi}{q}\langle \VEC{d}_{p,r}, \VEC{k} \rangle + Y_{p,r}
\end{equation*}
for a random variable $Y_{p,r}$ over $[-\pi, \pi)$ that satisfies
\begin{align*}
    \mathrm{Pr}(|Y_{p,r}| \geq \alpha) \leq \exp \left(-\frac{\eta \sin^2{(\alpha)}}{2} \text{SNR} \right)
\end{align*}
for $0 \leq \alpha \leq \pi/2$. Then,
we can write 
\begin{equation*}
    \text{arg}[U_{p,r} / U_{p}] = \frac{2 \pi}{q}\langle \VEC{e}_{r}, \VEC{k} \rangle + Y_{p, r} - Y_{p}.
\end{equation*}
Therefore,
\begin{align*}
    \mathrm{Pr}(Z_{p,r} \neq 0) &= \mathrm{Pr}(\text{arg}_q[U_{p,r} / U_{p}] \neq \langle \VEC{e}_{r}, \VEC{k} \rangle)\\
    &\leq \mathrm{Pr}(|Y_{p}| \geq \pi/2q) + \mathrm{Pr}(|Y_{p,r}| \geq \pi/2q)\\
    &\leq 2 \exp \left(-\frac{\eta \sin^2{(\pi/ 2q)}}{2} \text{SNR} \right).\hfill \qed
\end{align*}
\subsection{Proof of Theorem \ref{thm:noisy}}
\label{apdx:peeling_thm_noisy}

As stated in the theorem, we assume the alphabet size $q$ is a fixed constant throughout this proof. The success rate of the algorithm depends on each bin $\VEC{j}$ to be processed correctly, meaning that each bin is correctly identified as a zero-ton, singleton or multi-ton. Define $\mathcal{E}$ as the the error event where bin detector makes a mistake in $O(S)$ peeling iterations. If the error probability satisfies $\mathrm{Pr}(\mathcal{E}) \leq O(1/S)$, the probability of failure of peeling decoder can be written as
\begin{align*}
    \mathbb{P}_F &= \mathrm{Pr}\left(\widehat{F} \neq F | \mathcal{E}^{\mathrm{c}} \right) \mathrm{Pr}(\mathcal{E}^{\mathrm{c}}) + \mathrm{Pr}\left(\widehat{F} \neq F | \mathcal{E} \right) \mathrm{Pr}(\mathcal{E})\\
    &\leq \mathrm{Pr}\left(\widehat{F} \neq F | \mathcal{E}^{\mathrm{c}} \right) + \mathrm{Pr}(\mathcal{E})\\
    &= O(1/S)
\end{align*}
where the first term in the last inequality is obtained from Theorem \ref{thm:noiseless} for the peeling decoder with an oracle such that the event $\mathcal{E}^{\mathrm{c}}$ holds.

Then, we define $\mathcal{E}_b$ as the the error event where a bin $\VEC{j}$ is decoded wrongly and then, using a union bound over different bins and different iterations, the probability of the algorithm making a mistake in bin identification satisfies
\begin{equation*}
    \mathrm{Pr}(\mathcal{E}) \leq \text{(\# of iterations)} \times \text{(\# of bins)} \times \mathrm{Pr}(\mathcal{E}_b)
\end{equation*}

The number of bins is $\eta S$ and the number of iterations is at most $CS$ (at least one edge is peeled off at each iteration in the worst case). Hence, $\mathrm{Pr}(\mathcal{E}) \leq \eta C S^2 \mathrm{Pr}(\mathcal{E}_b)$. In order to satisfy $\mathrm{Pr}(\mathcal{E}) \leq O(1/S)$, we need to show that $\mathrm{Pr}(\mathcal{E}_b) \leq O(1/S^3)$.

In the following, we prove that $\mathrm{Pr}(\mathcal{E}_b) \leq O(1/S^3)$ holds using the observation model. In the following analysis, we consider separate cases where the bin in consideration is fixed as a zero-ton, singleton or multi-ton.

\begin{proposition}
The error probability $\mathrm{Pr}(\mathcal{E}_b)$ for an arbitrary bin can be upper bounded as $\mathrm{Pr}(\mathcal{E}_b) \leq O(1/S^3)$.
\end{proposition}

\begin{proof}
The error probability $\mathrm{Pr}(\mathcal{E}_b)$ for an arbitrary bin can be upper bounded as
\begin{align*}
    \mathrm{Pr}(\mathcal{E}_b) \leq &\sum_{\mathcal{F} \in \{\mathcal{H}_Z, \mathcal{H}_M\} } \mathrm{Pr}(\mathcal{F} \leftarrow \mathcal{H}_S(\VEC{k}, F[\VEC{k}]))\\
    + &\sum_{\mathcal{F} \in \{\mathcal{H}_Z, \mathcal{H}_M\} } \mathrm{Pr}(\mathcal{H}_S(\widehat{\VEC{k}}, \widehat{F}[\widehat{\VEC{k}}]) \leftarrow \mathcal{F})\\
    + &\mathrm{Pr}(\mathcal{H}_S(\widehat{\VEC{k}}, \widehat{F}[\widehat{\VEC{k}}]) \leftarrow \mathcal{H}_S(\VEC{k}, F[\VEC{k}]))   \end{align*}
where the events refer to:
\begin{enumerate}
    \item $\{ \mathcal{F} \leftarrow \mathcal{H}_S(\VEC{k}, F[\VEC{k}]) \}$:  missed verification in which the singleton verification fails when the ground truth is in fact a singleton.
    \item $\{\mathcal{H}_S(\widehat{\VEC{k}}, \widehat{F}[\widehat{\VEC{k}}]) \leftarrow \mathcal{F}\}$: false verification in which the singleton verification is passed when the ground truth is not a singleton.
    \item $\{ \mathcal{H}_S(\widehat{\VEC{k}}, \widehat{F}[\widehat{\VEC{k}}]) \leftarrow \mathcal{H}_S(\VEC{k}, F[\VEC{k}]) \}$: crossed verification in which a singleton with a wrong index-value pair passes the singleton verification when the ground truth is another singleton pair.
\end{enumerate}

    Since all the error probabilities decay exponentially with respect to $P_1$, it is clear that if $P_1$ is chosen as $P_1 = O(n) = O(\log N)$, the probability can be bounded as $\mathrm{Pr}(\mathcal{E}_b) \leq O(1/N^3)$. Then, the result follows by noting $S \leq N$.

Note that the number of samples required is $O(PB) = O(PS)$. For near-linear time bin detection, the total number of computations is dominated by the singleton search which requires $O(S N n) = O(S n q^n)$. On the other hand, for sublinear time bin detection, the total number of computations is dominated by the subsampling, which requires $O(PB\log B) = O(PS\log S) = O(PS\log N) $ because $S = O(N^\delta)$. As a result, this gives a total complexity of $O(PSn) = O(Sn^3)$. 
\end{proof}

\begin{proposition}[False Verification Rate]
    For $0 < \gamma < \frac{\eta}{2} \mathrm{SNR}$, the false verification rate for each bin hypothesis satisfies:
    \begin{align*}
        \mathrm{Pr}(\mathcal{H}_S(\widehat{\VEC{k}}, \widehat{F}[\widehat{\VEC{k}}]) \leftarrow \mathcal{H}_Z) &\leq e^{-\frac{P_1}{2}(\sqrt{1+2 \gamma} - 1)^2},\\
        \mathrm{Pr}(\mathcal{H}_S(\widehat{\VEC{k}}, \widehat{F}[\widehat{\VEC{k}}]) \leftarrow \mathcal{H}_M) &\leq e^{- \frac{P_1 \gamma^2}{2 (1 + 4\gamma)}} + 4 N^2 e^{-\epsilon \left( 1 - \frac{2 \gamma \nu^2}{L \rho^2} \right)^2 P_1},
    \end{align*}
where $P_1$ is the number of the random offsets in the near-linear time  and sublinear time robust bin detection algorithms.
\end{proposition}

\begin{proof}
 The probability of detecting a zero-ton as a singleton can be upper bounded by the probability of a zero-ton failing the zero-ton verification. Thus,
 \begin{align*}
     \mathrm{Pr}(\mathcal{H}_S(\widehat{\VEC{k}}, \widehat{F}[\widehat{\VEC{k}}]) \leftarrow \mathcal{H}_Z) &\leq  \mathrm{Pr}\left( \frac{1}{P_1} \|\VEC{W}\|^2 \geq (1+\gamma)\nu^2 \right)\\
     &\leq e^{-\frac{P_1}{2}(\sqrt{1+2 \gamma} - 1)^2}
 \end{align*}
by noting that $\VEC{W} \sim \mathcal{CN}(0, \nu^2 \VEC{I})$ and applying Lemma \ref{non-central-tail-bound}.

On the other hand, given some multi-ton observation $\VEC{U} = \VEC{S} \VEC{\alpha} + \VEC{W}$, the probability of detecting it as a singleton with index-value pair $(\widehat{\VEC{k}}, \widehat{F}[\widehat{\VEC{k}}] )$ can be written as
\begin{align*}
    \mathrm{Pr}(\mathcal{H}_S(\widehat{\VEC{k}}, \widehat{F}[\widehat{\VEC{k}}]) \leftarrow \mathcal{H}_Z) = \mathrm{Pr}\left( \frac{1}{P_1} \left\| \VEC{g} + \VEC{v} \right\|^2 \leq (1+\gamma)\nu^2\right),
\end{align*}
where $\VEC{g} := \VEC{S}(\VEC{\alpha} - \widehat{F}[\widehat{\VEC{k}}] \VEC{e}_{\widehat{\VEC{k}}})$ and $\VEC{v} := \VEC{W}$. Then, we can upper bound this probability by
\begin{align*}
    \mathrm{Pr}\left( \frac{1}{P_1} \left\| \VEC{g} + \VEC{v} \right\|^2 \leq (1+\gamma)\nu^2 \middle| \frac{\|\VEC{g}\|^2}{P_1} \geq 2 \gamma \nu^2 \right)\\ + \mathrm{Pr}\left( \frac{\|\VEC{g}\|^2}{P_1} \leq 2 \gamma \nu^2 \right).
\end{align*}

Using Lemma \ref{non-central-tail-bound}, the first term is upper bounded by $
    \exp\{- (P_1 \gamma^2)/(2 (1 + 4\gamma))\}$. To analyze the second term, we let $\VEC{\beta} = \VEC{\alpha} - \widehat{F}[\widehat{\VEC{k}}] \VEC{e}_{\widehat{\VEC{k}}}$ and write $\VEC{g} = \VEC{S} \VEC{\beta}$. Denoting its support as $\mathcal{L} := \mathrm{supp}(\VEC{\beta})$, we can further write $\VEC{S} \VEC{\beta} = \VEC{S}_{\mathcal{L}} \VEC{\beta}_{\mathcal{L}}$ where $\VEC{S}_{\mathcal{L}}$ is the sub-matrix 
 of $\VEC{S}$ consisting of the columns in $\mathcal{L}$ and $\VEC{\beta}_{\mathcal{L}}$ is the sub-vector consisting of the elements in $\mathcal{L}$. Then, we consider two scenarios:
\begin{itemize}
    \item The multi-ton size is a constant, i.e., $|\mathcal{L}| = L = O(1)$. In this case, we have
    \begin{equation*}
        \lambda_{\mathrm{min}}(\VEC{S}_{\mathcal{L}}^\mathrm{H} \VEC{S}_{\mathcal{L}}) \| \VEC{\beta}_{\mathcal{L}}\|^2 \leq \|\VEC{S}_{\mathcal{L}} \VEC{\beta}_{\mathcal{L}}\|^2
    \end{equation*}
    Using $\| \VEC{\beta}_{\mathcal{L}}\|^2 \geq L \rho^2$, the probability can be bounded as
    \begin{equation*}
        \mathrm{Pr}\left( \frac{\|\VEC{g}\|^2}{P_1} \leq 2 \gamma \nu^2 \right) \leq \mathrm{Pr}\left( \lambda_{\mathrm{min}} \left(\frac{1}{P_1} \VEC{S}_{\mathcal{L}}^\mathrm{H} \VEC{S}_{\mathcal{L}} \right) \leq \frac{2 \gamma \nu^2}{L \rho^2} \right)
    \end{equation*}
    Then, according to the Gershgorin Circle Theorem,
    \begin{equation*}
         \lambda_{\mathrm{min}}\left(\frac{1}{P_1} \VEC{S}_{\mathcal{L}}^\mathrm{H} \VEC{S}_{\mathcal{L}}\right) \geq 1 - L \mu,
    \end{equation*}
    where $\mu$ is the mutual coherence of $\VEC{S}$ defined in Lemma \ref{mutual_coherence}. Using the bound for $\mu$ and letting $\mu_0 = \frac{1}{L} \left( 1 - \frac{2 \gamma \nu^2}{L \rho^2} \right)$,
    \begin{align*}
        \mathrm{Pr}\left( \frac{\|\VEC{g}\|^2}{P_1} \leq 2 \gamma \nu^2 \right) \leq 4 N^2 e^{-\frac{1}{8L^2} \left( 1 - \frac{2 \gamma \nu^2}{L \rho^2} \right)^2 P_1}
    \end{align*}
    which holds as long as $\gamma < L \rho^2 / 2 \nu^2 = \frac{L \eta}{2} \mathrm{SNR}$.

    \item The multi-ton size grows asymptotically with respect to S, i.e., $|\mathcal{L}| = L = \omega(1)$. As a result, the vector of random variables $\VEC{g} = \VEC{S}_{\mathcal{L}} \VEC{\beta}_{\mathcal{L}}$ becomes asymptotically Gaussian due to the central limit theorem with zero mean and a covariance
    \begin{equation*}
        \mathbb{E}[\VEC{g} \VEC{g}^\mathrm{H}] = L \rho^2 \VEC{I}
    \end{equation*}
    Therefore, by Lemma \ref{non-central-tail-bound}, we have
    \begin{align*}
        \mathrm{Pr}\left( \frac{\|\VEC{g}\|^2}{P_1} \leq 2 \gamma \nu^2 \right) \leq e^{- \frac{P_1}{2} \left( 1 - \frac{2 \gamma \nu^2}{L \rho^2} \right)}
    \end{align*}
    which holds as long as $\gamma < L \rho^2 / 2 \nu^2 = \frac{L \eta}{2} \mathrm{SNR}$.
\end{itemize}    

By combining the results from both cases, there exists some absolute constant $\epsilon > 0$ such that 
\begin{align*}
    \mathrm{Pr}\left( \frac{\|\VEC{g}\|^2}{P_1} \leq 2 \gamma \nu^2 \right) \leq 4 N^2 e^{-\epsilon \left( 1 - \frac{2 \gamma \nu^2}{L \rho^2} \right)^2 P_1}
\end{align*}
as long as $\gamma < \rho^2 / 2 \nu^2 = \frac{\eta}{2} \mathrm{SNR}$.
\end{proof}

\begin{proposition}[Missed Verification Rate]
    For $0 < \gamma < \frac{\eta}{2} \mathrm{SNR}$, the missed verification rate for each bin hypothesis satisfies
    \begin{align*}
        &\mathrm{Pr}( \mathcal{H}_Z \leftarrow \mathcal{H}_S(\VEC{k}, F[\VEC{k}])) \leq e^{-\frac{P_1}{2}\frac{(\rho^2/\nu^2 - \gamma)^2}{1 + 2\rho^2/\nu^2}}\\
        &\mathrm{Pr}(\mathcal{H}_M \leftarrow \mathcal{H}_S(\VEC{k}, F[\VEC{k}]) )\\
        &\quad \leq e^{-\frac{P_1}{2}(\sqrt{1+2 \gamma} - 1)^2} + 2 e^{- \frac{\rho^2 \sin^2(\pi/\kappa)}{2 \nu^2} P_1} \\
        &\quad + \begin{cases}
            4 N e^{-\frac{P_1}{2} \left( \sqrt{2 \gamma + 1} - 1\right)^2}\\
            \quad + 4 N^2 e^{-\frac{P_1}{8}\left( 1 - \sqrt{16 (1+\gamma) \nu^2/\rho^2 } \right)^2}, &\text{NLT} \\
             4 n e^{-\frac{1}{2 q} \epsilon^2 P_1}, &\text{SLT}
        \end{cases}
    \end{align*}
where $P_1$ is the number of the random offsets in the near-linear time  and sublinear time robust bin detection algorithms.
\end{proposition}

\begin{proof}
    The probability of detecting a singleton as a zero-ton can be upper bounded by the probability of a singleton passing the zero-ton verification. Hence, by noting that $\VEC{W} \sim \mathcal{CN}(0, \nu^2 \VEC{I})$ and applying Lemma \ref{non-central-tail-bound},
\begin{align*}
     \mathrm{Pr}&( \mathcal{H}_Z \leftarrow \mathcal{H}_S(\VEC{k}, F[\VEC{k}])) \\
     &\leq  \mathrm{Pr}\left( \frac{1}{P_1} \|F[\VEC{k}] \VEC{s}_{\VEC{k}} + \VEC{W}\|^2 \leq (1+\gamma)\nu^2 \right)\\
     &\leq e^{-\frac{P_1}{2}\frac{(\rho^2/\nu^2 - \gamma)^2}{1 + 2\rho^2/\nu^2}}.
 \end{align*}
which holds as long as $\gamma < \rho^2 / \nu^2 = \eta \mathrm{SNR}$.

On the other hand, the probability of detecting a singleton as a multi-ton can be written as the probability of failing the singleton verification step for some index-value pair $(\widehat{\VEC{k}}, \widehat{F}[\widehat{\VEC{k}}] )$. Denoting the error event by 
\begin{equation*}
\mathcal{E}_{V} = \left\{ \frac{1}{P} \left\| \VEC{U} - \widehat{F}[\widehat{\VEC{k}}] \VEC{s}_{\widehat{k}} \right\|^2 \geq (1 + \gamma) \nu^2 \right\}    ,
\end{equation*} 
we can write, 
\begin{align*}
    \mathrm{Pr}(\mathcal{H}_M \leftarrow \mathcal{H}_S(\VEC{k}, F[\VEC{k}]) ) &= \mathrm{Pr} (\mathcal{E}_{V}) \\
    &\leq \mathrm{Pr} (\mathcal{E}_{V} | \widehat{F}[\widehat{\VEC{k}}] = F[\VEC{k}] \wedge \widehat{\VEC{k}} = \VEC{k})\\
    &+ \mathrm{Pr} (\widehat{F}[\widehat{\VEC{k}}] \neq F[\VEC{k}]  \vee \widehat{\VEC{k}} \neq \VEC{k}).
\end{align*}

Here, the first term is bounded as
\begin{align*}
    \mathrm{Pr} (\mathcal{E}_{V} | \widehat{F}[\widehat{\VEC{k}}] = F[\VEC{k}] \wedge \widehat{\VEC{k}} = \VEC{k}) &\leq  \mathrm{Pr}\Big( \frac{1}{P_1} \|\VEC{W}\|^2 \geq (1+\gamma)\nu^2 \Big)\\
     &\leq e^{-\frac{P_1}{2}(\sqrt{1+2 \gamma} - 1)^2}.
\end{align*}

The second term is bounded as
\begin{align*}
    \mathrm{Pr} (\widehat{F}[\widehat{\VEC{k}}] &\neq F[\VEC{k}]  \vee \widehat{\VEC{k}} \neq \VEC{k}) \\
    &\leq \mathrm{Pr} (\widehat{F}[\widehat{\VEC{k}}] \neq F[\VEC{k}]) + \mathrm{Pr} (\widehat{\VEC{k}} \neq \VEC{k})\\
    &= \mathrm{Pr} (\widehat{F}[\widehat{\VEC{k}}] \neq F[\VEC{k}] | \widehat{\VEC{k}} \neq \VEC{k}) \mathrm{Pr} (\widehat{\VEC{k}} \neq \VEC{k})\\
    &\qquad + \mathrm{Pr} (\widehat{F}[\widehat{\VEC{k}}] \neq F[\VEC{k}] | \widehat{\VEC{k}} = \VEC{k}) \mathrm{Pr} (\widehat{\VEC{k}} = \VEC{k})\\
    &\qquad + \mathrm{Pr} (\widehat{\VEC{k}} \neq \VEC{k})\\
    &\leq \mathrm{Pr} (\widehat{F}[\widehat{\VEC{k}}] \neq F[\VEC{k}] | \widehat{\VEC{k}} = \VEC{k}) + 2 \mathrm{Pr} (\widehat{\VEC{k}} \neq \VEC{k})
\end{align*}

The first term is the error probability of a $\kappa$-point PSK signal with constellation points $\{ \rho, \rho \phi, \rho \phi^2, \dots, \rho \phi^{\kappa-1} \}$, and can be bounded as
\begin{equation*}
    \mathrm{Pr} (\widehat{F}[\widehat{\VEC{k}}] \neq F[\VEC{k}] | \widehat{\VEC{k}} = \VEC{k}) \leq 2 e^{- \frac{\rho^2 \sin^2(\pi/\kappa)}{2 \nu^2} P_1}
\end{equation*}

Since the second term $\mathrm{Pr} (\widehat{\VEC{k}} \neq \VEC{k})$ is the error probability of the singleton search, we use Lemmas \ref{linear_identification_error} and \ref{sublinear_identification_error}.
\end{proof}

\begin{proposition}[Crossed Verification Rate]
For $0 < \gamma < \frac{\eta}{2} \mathrm{SNR}$, the crossed verification rate for each bin hypothesis satisfies
\begin{align*}
    &\mathrm{Pr}(\mathcal{H}_S(\widehat{\VEC{k}}, \widehat{F}[\widehat{\VEC{k}}]) \leftarrow \mathcal{H}_S(\VEC{k}, F[\VEC{k}])) \\
    &\qquad \leq 4 N^2 e^{-\frac{1}{32} \left( 1 - \frac{2 \gamma \nu^2}{2 \rho^2} \right)^2 P_1} + e^{- \frac{P_1 \gamma^2}{2 (1 + 4\gamma)}}
\end{align*}
where $P_1$ is the number of the random offsets in the near-linear time  and sublinear time robust bin detection algorithms.
\end{proposition}

\begin{proof}
This error event can only occur if a singleton with index-value pair $(\VEC{k}, F[\VEC{k}])$ passes the singleton verification step for some index-value pair $(\widehat{\VEC{k}}, \widehat{F}[\widehat{\VEC{k}}] )$ such that $\VEC{k} \neq \widehat{\VEC{k}}$. Hence,
\begin{align*}
    &\mathrm{Pr}(\mathcal{H}_S(\widehat{\VEC{k}}, \widehat{F}[\widehat{\VEC{k}}]) \leftarrow \mathcal{H}_S(\VEC{k}, F[\VEC{k}])) \\
    &\leq  \mathrm{Pr}\left( \frac{1}{P_1} \|F[\VEC{k}] \VEC{s}_{\VEC{k}} - \widehat{F}[\widehat{\VEC{k}}] \VEC{s}_{\widehat{\VEC{k}}} + \VEC{W}\|^2 \leq (1+\gamma)\nu^2 \right)\\
    &= \mathrm{Pr}\left( \frac{1}{P_1} \|\VEC{S} \VEC{\beta} + \VEC{W}\|^2 \leq (1+\gamma)\nu^2 \right)\\
    &= \mathrm{Pr}\left( \frac{1}{P_1} \|\VEC{S} \VEC{\beta} + \VEC{W}\|^2 \leq (1+\gamma)\nu^2 \middle | \|\VEC{S} \VEC{\beta}\|^2 \geq 2 \gamma \nu^2 \right)\\
    &\qquad + \mathrm{Pr}\left( \|\VEC{S} \VEC{\beta}\|^2 \leq 2 \gamma \nu^2 \right)
\end{align*}
where $\VEC{\beta}$ is a $2$-sparse vector with non-zero entries from $\{ \rho, \rho \phi, \rho \phi^2, \dots, \rho \phi^{\kappa-1} \}$. Using Lemma \ref{non-central-tail-bound}, the first term is upper-bounded by $e^{- \frac{P_1 \gamma^2}{2 (1 + 4\gamma)}}$. By Lemma \ref{mutual_coherence}, the second term is upper bounded by $4 N^2 e^{-\frac{1}{32} \left( 1 - \frac{2 \gamma \nu^2}{2 \rho^2} \right)^2 P_1}$. 
\end{proof}

\begin{lemma}[Non-central Tail Bounds]
    Given $\VEC{g} \in \mathbb{C}^P$ and a complex Gaussian vector $\VEC{v} \sim \mathcal{CN}(0, \nu^2 \VEC{I})$, the following tail bounds hold:
    \begin{align*}
        \mathrm{Pr}\left( \frac{1}{P}\|\VEC{g} + \VEC{v}\|^2 \geq \tau_1 \right) &\leq e^{-\frac{P}{2} (\sqrt{2 \tau_1/\nu^2-1} - \sqrt{1 + 2\theta_0})^2}\\
        \mathrm{Pr}\left( \frac{1}{P}\|\VEC{g} + \VEC{v}\|^2 \leq \tau_2 \right) &\leq e^{-\frac{P}{2} \frac{\left(1+\theta_0-\tau_2/\nu^2 \right)^2}{1+2\theta_0}}\\
    \end{align*}
for any $\tau_1$ and $\tau_2$ that satisfy
\begin{equation*}
    \tau_1 \geq \nu^2 (1 + \theta_0) \geq \tau_2,
\end{equation*}
where $\theta_0$ is the normalized non-centrality parameter given by 
\begin{equation*}
    \theta_0 := \frac{\|\VEC{g}\|^2}{P \nu^2}.
\end{equation*}
\label{non-central-tail-bound}
\end{lemma}

\begin{proof}
    We apply Lemma 11 in \cite{Li2015} to $\widetilde{\VEC{g}} = [\mathrm{Re}[\VEC{g}], \mathrm{Im}[\VEC{g}]]^\mathrm{T} \in \mathbb{R}^{2P}$ and $\widetilde{\VEC{v}} = [\mathrm{Re}[\VEC{v}], \mathrm{Im}[\VEC{v}]]^\mathrm{T} \sim \mathcal{N}(0, \frac{\nu^2}{2} \VEC{I})$.
\end{proof}

\begin{lemma}
    Denote the mutual coherence of $\VEC{S}$ by 
    \begin{equation*}
     \mu := \max_{\VEC{k} \neq \VEC{m}} \frac{1}{P_1} |\VEC{s}_{\VEC{k}}^{\mathrm{H}} \VEC{s}_{\VEC{m}}| . 
    \end{equation*}
    Then, we have $\mathrm{Pr}(\mu \geq \mu_0) \leq 4 N^2 e^{- \frac{\mu_0^2}{8} P_1}$ for any $\mu_0 > 0$.
    \label{mutual_coherence}
\end{lemma}

\begin{proof}
The columns of $\VEC{S}$ are given as $\VEC{s}_{\VEC{k}} = \omega^{\VEC{D} \VEC{k}}$. Let us denote the $p$-th entry of $\VEC{s}_{\VEC{k}}$ by $\VEC{s}_{\VEC{k}}[p] = \omega^{\langle\VEC{d}_{p}, \VEC{k}\rangle}$ for $p \in [P_1]$. Therefore, each multiplication can be written as $\VEC{y}[p] := \VEC{s}_{\VEC{k}}[p]^{\mathrm{H}} \VEC{s}_{\VEC{m}}[p] = \omega^{\langle\VEC{d}_{p}, \VEC{m} \ominus_q \VEC{k}\rangle}$. Since $\VEC{m} \ominus_q \VEC{k} \neq \VEC{0}$ for all $\VEC{k} \neq \VEC{m}$, each $\VEC{y}[p]$ term is uniformly distributed over the set $ \{ \rho, \rho \omega, \rho \omega^2, \dots, \rho \omega^{q-1} \}$. Then, due to independence in choosing offsets $\VEC{d}_{p}$, we can apply a Hoeffding bound to obtain
\begin{align*}
    \mathrm{Pr} \left(\frac{|\VEC{s}_{\VEC{k}}^{\mathrm{H}} \VEC{s}_{\VEC{m}}|}{P_1}  \geq \mu_0 \right) &= \mathrm{Pr}(|\VEC{y}| \geq P_1 \mu_0) \\
    &\leq 2 \mathrm{Pr}(\mathrm{Re}|\VEC{y}| \geq P_1 \mu_0 / \sqrt{2})\\
    &\leq 4 e^{- \frac{\mu_0^2}{8} P_1}.
\end{align*}

By applying a union bound over all $(\VEC{k}, \VEC{m})$ pairs, we obtain the result.
\end{proof}

\begin{lemma}
    The singleton search error probability of the near-linear time robust bin detection is upper bounded as
    \begin{equation*}
        \mathrm{Pr} (\widehat{\VEC{k}} \neq \VEC{k}) \leq 4 N e^{-\frac{P}{2} \left( \sqrt{2 \epsilon + 1} - 1\right)^2} + 4 N^2 e^{-\frac{P}{8}\left( 1 - \sqrt{\frac{16 (1+\epsilon) \nu^2} {\rho^2} } \right)^2}
    \end{equation*}
for some constant $\epsilon > 0$.
\label{linear_identification_error}
\end{lemma}

\begin{proof}
The error occurs if the residual energy for some $\VEC{m} \neq \VEC{k}$ is lower than the residual energy for true singleton index $\VEC{k}$. Let $\mathcal{H}$ denote the event $\{ \mu < 1 - \sqrt{16 (1 + \epsilon) \nu^2 / \rho^2 } \}$ where $\mu$ is the mutual coherence of $\VEC{S}$ and $\gamma$ is a constant. Furthermore, let $\mathrm{Pr}_{\mathcal{H}}$ denote probabilities conditioned on $\mathcal{H}$. Then, for any given $\VEC{m} \neq \VEC{k}$,
\begin{align*}
    &\mathrm{Pr}_{\mathcal{H}} (\widehat{\VEC{k}} = \VEC{m}) \\
    &\qquad \leq \mathrm{Pr}_{\mathcal{H}} \Bigg( \left\|\frac{1}{P} \VEC{s}_{\VEC{m}}^\mathrm{T} \VEC{U} \VEC{s}_{\VEC{m}} - \VEC{U} \right\| \leq \left\|\frac{1}{P} \VEC{s}_{\VEC{k}}^\mathrm{T} \VEC{U} \VEC{s}_{\VEC{k}} - \VEC{U} \right\| \Bigg) \\
    &\quad = \mathrm{Pr}_{\mathcal{H}} \Bigg( \left\|\frac{1}{P} \VEC{s}_{\VEC{m}}^\mathrm{T} \left( F[\VEC{k}] \VEC{s}_{\VEC{k}} + \VEC{W} \right) \VEC{s}_{\VEC{m}} - F[\VEC{k}] \VEC{s}_{\VEC{k}} - \VEC{W} \right\| \\
    &\quad \qquad \qquad \leq \left\|\frac{1}{P} \VEC{s}_{\VEC{k}}^\mathrm{T} \left(F[\VEC{k}] \VEC{s}_{\VEC{k}} + \VEC{W} \right) \VEC{s}_{\VEC{k}} - F[\VEC{k}] \VEC{s}_{\VEC{k}} - \VEC{W} \right\| \Bigg) \\
    &\quad = \mathrm{Pr}_{\mathcal{H}} \Bigg( \left\| F[\VEC{k}] \left( \frac{1}{P} \VEC{s}_{\VEC{m}}^\mathrm{T} \VEC{s}_{\VEC{k}} \VEC{s}_{\VEC{m}} - \VEC{s}_{\VEC{k}} \right) + \frac{1}{P} \VEC{s}_{\VEC{m}}^\mathrm{T} \VEC{W} \VEC{s}_{\VEC{m}} - \VEC{W} \right\| \\
    &\quad \qquad \qquad \leq \left\|\frac{1}{P} \VEC{s}_{\VEC{k}}^\mathrm{T} \VEC{W} \VEC{s}_{\VEC{k}} - \VEC{W} \right\| \Bigg).
\end{align*}
Then, 
\begin{align*}
    &\mathrm{Pr}_{\mathcal{H}} (\widehat{\VEC{k}} = \VEC{m}) \\
    &\qquad \leq \mathrm{Pr}_{\mathcal{H}} \Bigg( \rho \left\| \frac{1}{P} \VEC{s}_{\VEC{m}}^\mathrm{T} \VEC{s}_{\VEC{k}} \VEC{s}_{\VEC{m}} - \VEC{s}_{\VEC{k}} \right\| - \left\| \frac{1}{P} \VEC{s}_{\VEC{m}}^\mathrm{T} \VEC{W} \VEC{s}_{\VEC{m}} - \VEC{W} \right\| \\
    &\qquad \qquad \qquad \leq \left\|\frac{1}{P} \VEC{s}_{\VEC{k}}^\mathrm{T} \VEC{W} \VEC{s}_{\VEC{k}} - \VEC{W} \right\| \Bigg) \\
&\qquad \leq 4 \mathrm{Pr}_{\mathcal{H}} \Bigg( \frac{1}{P} \left\| \VEC{W} \right\|^2 \geq (1 + \epsilon) \nu^2 \Bigg)\\
    &\qquad \leq 4 e^{-\frac{P}{2} \left( \sqrt{2 \epsilon + 1} - 1\right)^2}.
\end{align*}
Next, by a union bound over all $\VEC{m} \in \mathbb{Z}_q^n$, we have 
\begin{equation*}
    \mathrm{Pr}_{\mathcal{H}} (\widehat{\VEC{k}} = \VEC{k}) \leq 4 N e^{-\frac{P}{2} \left( \sqrt{2 \epsilon + 1} - 1\right)^2}.
\end{equation*}
Lastly, we can write
\begin{align*}
    \mathrm{Pr} (\widehat{\VEC{k}} &= \VEC{k}) \leq \mathrm{Pr} (\widehat{\VEC{k}} = \VEC{k} | \mathcal{H}) + \mathrm{Pr}(\mathcal{H}^\mathrm{c})\\
    &\leq 4 N e^{-\frac{P}{2} \left( \sqrt{2 \epsilon + 1} - 1\right)^2} + 4 N^2 e^{-\frac{P}{8}\left( 1 - \sqrt{\frac{16 (1+\epsilon) \nu^2} {\rho^2} } \right)^2}.
\end{align*}
\end{proof}

\begin{lemma}
The singleton search error probability of the sub-linear time robust bin detection is upper bounded as
    \begin{equation*}
        \mathrm{Pr} (\widehat{\VEC{k}} \neq \VEC{k}) \leq 2 n e^{-\frac{1}{2 q} \epsilon^2 P_1}
    \end{equation*}
for some constant $\epsilon > 0$.
\label{sublinear_identification_error}
\end{lemma}

\begin{proof}
Recall that $\text{arg}_q[U_{p,r} / U_{p}] = k[r] \oplus_q Z_{p,r}$ where $Z_{p,r}$ is a random variable over $\mathbb{Z}_q$ with parameters $(p_0, p_1, \dots, p_{q-1})$ such that $p_0 - p_i \geq \epsilon$ for all $i\neq0$ for some constant $\epsilon > 0$. Then, the error probability of the singleton search for the $r$-th symbol of $\VEC{k}$ is
\begin{align*}
    \mathrm{Pr} &(\widehat{k}[r] \neq k[r]) \\ 
    &= \mathrm{Pr}  \Big( \arg \max_{a \in \mathbb{Z}_q} \sum_{p \in [P_1]} \mathbb{1}\{\text{arg}_q[U_{p,r} / U_{p}] =  a\} \neq k[r] \Big)\\
    &= \mathrm{Pr}  \Big( \arg \max_{a \in \mathbb{Z}_q} \sum_{p \in [P_1]} \mathbb{1}\{Z_{p,r} = a \} \neq 0 \Big)\\
    &\leq \mathrm{Pr}  \Big( \|\VEC{p} - \widehat{\VEC{p}}\|_1 \geq \epsilon \Big)\\
    &\leq 2 e^{-\frac{1}{2 q} \epsilon^2 P_1}.
\end{align*}
where $\VEC{p} = [p_0, p_1, \dots, p_{q-1}]^\mathrm{T}$ is a vector, $\widehat{\VEC{p}}$ is a random vector with distribution $\frac{1}{n} \textrm{Multinomial}(P_1, \VEC{p})$ and the last step uses the result from \cite{weissman2003}. By union bounding over all $n$ symbols, we have the final result.
\end{proof}

\newpage
\bibliographystyle{IEEEtran}
\bibliography{IEEEabrv,references}

\end{document}